\documentclass[journal,a4paper]{IEEEtran}
\pdfoutput=1

\usepackage[utf8]{inputenc}
\usepackage[english]{babel}
\usepackage[T1]{fontenc}
\usepackage{epsfig}
\usepackage{amsmath, amssymb, amsbsy}
\usepackage{mathdots}
\usepackage{xspace}
\usepackage[noend]{algpseudocode}
\usepackage[linesnumbered,ruled,vlined,titlenumbered]{algorithm2e}

\usepackage{color}
\usepackage{cite}
\usepackage{booktabs}
\usepackage{verbatim}
\usepackage[colorlinks=true,citecolor=blue,linkcolor=blue,urlcolor=blue]{hyperref}
\usepackage{url}
\usepackage{lipsum}
\usepackage{pgfplots}
\pgfplotsset{compat=newest}
\usepackage{tikz}
\usetikzlibrary{arrows,matrix,positioning,shapes}
\usepackage{enumitem}
\usepackage{colortbl}
\usepackage{tablefootnote}
\usepackage[nice]{nicefrac}

\usepackage[amsmath,thmmarks,thref]{ntheorem}
\theorempreskip{.5\baselineskip}
\theorempostskip{.5\baselineskip}
\newtheorem{theorem}{Theorem}
\newtheorem{lemma}[theorem]{Lemma}
\newtheorem{corollary}[theorem]{Corollary}
\newtheorem{proposition}[theorem]{Proposition}

\newtheorem{problem}[theorem]{Problem}
\newtheorem{example}[theorem]{Example}
\newtheorem{remark}[theorem]{Remark}
\newtheorem*{expectation}[theorem]{Expectation}
\theorembodyfont{\normalfont}
\newtheorem{definition}{Definition}

\usepackage[final,tracking=true,kerning=true,spacing=true,factor=1100,stretch=10,shrink=20]{microtype}

\algrenewcommand\alglinenumber[1]{{\scriptsize#1}}
\algrenewcommand\algorithmicrequire{\textbf{Input:}}
\algrenewcommand\algorithmicensure{\textbf{Output:}}

\usepackage{varioref}        %
\usepackage{fancyref}

\makeatletter
\def\mkfancyprefix#1#2{%
	\expandafter\def\csname fancyref#1labelprefix\endcsname{#1}%
	\begingroup\def\x{\endgroup\frefformat{plain}}%
	\expandafter\x\csname fancyref#1labelprefix\endcsname
	{\MakeLowercase{#2}\fancyrefdefaultspacing##1}%
	\begingroup\def\x{\endgroup\Frefformat{plain}}%
	\expandafter\x\csname fancyref#1labelprefix\endcsname
	{#2\fancyrefdefaultspacing##1}%
	\begingroup\def\x{\endgroup\frefformat{vario}}%
	\expandafter\x\csname fancyref#1labelprefix\endcsname
	{\MakeLowercase{#2}\fancyrefdefaultspacing##1##3}%
	\begingroup\def\x{\endgroup\Frefformat{vario}}%
	\expandafter\x\csname fancyref#1labelprefix\endcsname
	{#2\fancyrefdefaultspacing##1##3}%
}
\makeatother
\fancyrefchangeprefix{\fancyrefeqlabelprefix}{eqn}
\mkfancyprefix{ssec}{Section}
\mkfancyprefix{tbl}{Table}
\mkfancyprefix{thm}{Theorem}
\mkfancyprefix{lem}{Lemma}
\mkfancyprefix{cor}{Corollary}
\mkfancyprefix{prop}{Proposition}
\mkfancyprefix{prob}{Problem}
\mkfancyprefix{alg}{Algorithm}
\mkfancyprefix{inv}{Invariant}
\mkfancyprefix{ex}{Example}
\mkfancyprefix{line}{Line}
\mkfancyprefix{def}{Definition}
\mkfancyprefix{itm}{Item}
\mkfancyprefix{app}{Appendix}
\mkfancyprefix{rem}{Remark}
\newcommand{\cref}[1]{\Fref{#1}}

\def\ve#1{{\mathchoice{\mbox{\boldmath$\displaystyle #1$}}%
		{\mbox{\boldmath$\textstyle #1$}}%
		{\mbox{\boldmath$\scriptstyle #1$}}%
		{\mbox{\boldmath$\scriptscriptstyle #1$}}}}

\definecolor{sunset}{rgb}{1,0.5,.05}
\definecolor{marine}{rgb}{0,0,.7}
\definecolor{navy}{rgb}{0,0,.5}
\definecolor{forest}{rgb}{0,.6,0}
\definecolor{brown}{rgb}{0.59, 0.29, 0.0}

\newcommand{\mo}{{-1}}

\newcommand{\F}{\mathbb{F}}

\newcommand{\ZZ}{\mathbb{Z}}

\newcommand{\M}{{\cal{M}}}
\newcommand{\m}{\ensuremath{\ve{m}}}

\renewcommand{\v}{\ensuremath{\ve{v}}}

\renewcommand{\vec}[1]{\ensuremath{\ve{#1}}}

\newcommand{\NN}{\mathbb{N}}

\DeclareMathOperator{\rank}{rank}
\DeclareMathOperator{\diag}{diag}

\AtBeginDocument{}

\renewcommand{\c}{\ve{c}}
\newcommand{\e}{\ve{e}}
\renewcommand{\r}{\ve{r}}

\newcommand{\Code}{\mathcal{C}}
\newcommand{\wtH}{\mathrm{wt_H}}
\newcommand{\dH}{\mathrm{d_H}}

\newcommand{\Fq}{\mathbb{F}_q}

\newcommand{\numTwists}{\ell}

\newcommand{\tVec}{{\ve t}}
\newcommand{\hVec}{{\ve h}}
\newcommand{\etaVec}{{\ve \eta}}
\newcommand{\alphaVec}{{\ve \alpha}}

\newcommand{\evpolys}{\mathcal{P}^{n,k}_{\tVec,\hVec,\etaVec}}
\newcommand{\twisted}{$[k,\tVec,\hVec,\etaVec]$-twisted }
\newcommand{\twistedC}{$[\alphaVec,\tVec,\hVec,\etaVec]$-twisted }

\DeclareMathOperator{\evOp}{ev}
\newcommand{\ev}[2]{\evOp_{#2}\!\left(#1\right)}
\newcommand{\TRS}[2]{\Code_{#1}^{#2}}
\newcommand{\Cmult}{{\Code_{\alphaVec,\tVec,\hVec,\etaVec}^{n,k}}}

\newcommand{\Iset}{\mathcal{I}}

\newcommand{\Bmat}{\ve B}
\newcommand{\Gmat}{\ve G}
\newcommand{\Hmat}{\ve H}
\newcommand{\Lmat}{\ve L}
\newcommand{\J}{{\ve J}}
\newcommand{\Van}{\vec V}
\newcommand{\I}{\ve I}
\newcommand{\A}{\ve A}
\newcommand{\B}{\ve B}
\newcommand{\Amat}{\ve A}
\newcommand{\Dmat}{\ve D}
\newcommand{\x}{\ve x}

\newcommand{\mymat}[1]{}

\definecolor{etacolor1}{rgb}{0,0,0.7}
\definecolor{etacolor2}{rgb}{0,0.7,0}
\definecolor{etacolor3}{rgb}{0.63922,0.14902,0.21961}

\newcommand{\dual}{\bot}

\newcommand{\etaSet}{\mathcal{H}}

\newcommand{\pij}[1]{p^{(#1)}}

\newcommand{\AetaV}{A^{(\etaVec)}}

\newcommand{\AetaM}{\A^{(\etaVec)}}
\newcommand{\Atilde}{\ve{A}'}

\newcommand{\DsetText}[1]{\mathrm{D}(#1)}
\newcommand{\Dsetsmallern}[1]{\mathrm{D}\!\left(#1\right)_{<n}}
\newcommand{\DsetsmallernText}[1]{\mathrm{D}(#1)_{<n}}
\newcommand{\Dbarset}[1]{\overline{\mathrm{D}}\!\left(#1\right)}

\newcommand{\Pset}{\mathcal{P}}

\newcommand{\startw}{$(*)$-twisted }

\newcommand{\plustw}{$(+)$-twisted }

\renewcommand{\l}{{\ve l}}

\newcommand{\0}{\ve{0}}

\newcommand{\Fqsmall}{\mathbb{F}_{q_0}}
\DeclareMathOperator{\mymod}{mod}
\newcommand{\rem}{\revision{\mymod}}

\renewcommand{\i}{\ve{i}}
\renewcommand{\j}{\ve{j}}
\renewcommand{\e}{\ve{e}}
\newcommand{\deltamu}{\ve{\delta}_\mu}
\newcommand{\Eset}{\mathcal{E}}
\DeclareMathOperator{\supp}{supp}

\newcommand{\decParam}{\zeta}

\newcommand{\ModuleBasis}{\ve{M}}
\newcommand{\Module}{\mathcal{M}}

\newcommand{\s}{\ve{s}}

\newcommand{\Rmat}{{\ve R}}

\newcommand{\tauLB}{\tau_\mathsf{LB}}

\definecolor{mytablecellcolor}{rgb}{0.85,0.85,0.85}
\newcommand{\colorthiscell}{\cellcolor{mytablecellcolor}}

\newcommand{\PfTwoBelow}{\mathrm{P}_{\tau_\mathsf{max}-1}^\mathsf{max}}
\newcommand{\PfOneBelow}{\mathrm{P}_{\tau_\mathsf{max}}^\mathsf{max}}
\newcommand{\PfAbove}{\mathrm{P}_{\tau_\mathsf{max}+1}^\mathsf{min}}
\newcommand{\mysubsubsection}[1]{\subsubsection{#1}}

\newcommand{\taumax}{\tau_\mathsf{max}}

\definecolor{darkgreen}{rgb}{0,0.7,0}

\definecolor{darkred}{rgb}{0.7,0,0}

\newcommand{\revision}[1]{{\color{blue}{#1}}}

\begin{document}

\title{Twisted Reed--Solomon Codes}
\author{\IEEEauthorblockN{Peter Beelen, Sven Puchinger,~\IEEEmembership{Member,~IEEE}, and Johan Rosenkilde}\\
\thanks{Parts of this work have been presented at the IEEE International Symposium on Information Theory (ISIT) 2017 \cite{beelen2017twisted} and 2018 \cite{beelen2018structural}.}
\thanks{P.~Beelen is with the Department of Applied Mathematics and Computer Science, Technical University of Denmark, 2800 Kongens Lyngby, Denmark (e-mail: pabe@dtu.dk).}
\thanks{S.~Puchinger is with Hensoldt Sensors GmbH, 89077 Ulm, Germany (e-mail: mail@svenpuchinger.de). This work was done while he was with the Department of Applied Mathematics and Computer Science, Technical University of Denmark (DTU), 2800 Kongens Lyngby, Denmark, and the Department of Electrical and Computer Engineering, Technical University of Munich, 80333 Munich, Germany.}
\thanks{J.~Rosenkilde is with GitHub Denmark Aps, 2100 Copenhagen, Denmark (email: jsrn@jsrn.dk). Most of this work was done while he was with the Department of Applied Mathematics and Computer Science, Technical University of Denmark (DTU), 2800 Kongens Lyngby, Denmark.}
\thanks{The first and third author would like to acknowledge the support from The Danish Council for Independent Research (DFF-FNU) for the project \emph{Correcting on a Curve}, Grant No.~8021-00030B. The second author was supported by the European Union's Horizon 2020 research and innovation program under the Marie Sklodowska-Curie grant agreement no.~713683 and by the European Research Council (ERC) under the European Union’s Horizon 2020 research and innovation programme (grant agreement no.~801434).}
}

\maketitle

\begin{abstract}
In this article, we present a new construction of evaluation codes in the Hamming metric, which we call \emph{twisted Reed--Solomon codes}. Whereas Reed--Solomon (RS) codes are MDS codes, this need not be the case for twisted RS codes. Nonetheless, we show that our construction yields several families of MDS codes. 
Further, for a large subclass of (MDS) twisted RS codes, we show that the new codes are not generalized RS codes.
To achieve this, we use properties of Schur squares of codes as well as an explicit description of the dual of a large subclass of our codes. We conclude the paper with a description of a decoder, that performs very well in practice as shown by extensive simulation results.
\end{abstract}


\section{Introduction}

Maximum distance separable (MDS) codes are error-correcting codes with a particularly large minimum distance. More precisely, they are linear $[n,k,d]$ codes over a finite field $\mathbb{F}_q$ where $d=n-k+1$, i.e., meeting the Singleton bound. The well known family of generalized Reed--Solomon (GRS) codes are MDS codes, thus giving examples of MDS codes of length up to $q+1$.
Other known MDS codes have been constructed from $n$-arcs in projective geometry \cite{macwilliams1977theory}, circulant matrices \cite{roth1989mds}, or Hankel matrices \cite{roth1989mds}.
In this paper, we consolidate and extend the study of twisted Reed--Solomon (twisted RS) codes initiated in the conference papers \cite{beelen2017twisted,beelen2018structural}.
This new code family is inspired by Sheekey's twisted Gabidulin codes \cite{sheekey2015new}, a class of rank-metric codes.
The class of twisted RS codes contains several subfamilies of long MDS codes.

More precisely, after giving some needed preliminaries in the second section, we introduce the class of twisted RS codes in Section~\ref{sec:twisted_rs_codes}. After this, we study several special cases in the fourth section that give rise to MDS codes.
In Section~\ref{sec:duals}, we present results on the duals of twisted RS-codes.
We compare twisted RS codes with GRS codes in the sixth section. We are able to give various families of MDS twisted RS codes that are not monomially equivalent to GRS codes. Our main tool for this will be the Schur square of a code, which has low dimension for a GRS code, but can have a large dimension for a twisted RS code. %
In the last section, we discuss decoding of twisted RS codes and indicate a decoder that works very well in practice.

While working on this paper, related papers on twisted RS codes have begun to appear. In \cite{liu2020new}, a construction was presented that can give slightly longer MDS twisted RS codes by modifying two of our special classes of MDS twisted RS codes. %
Further, one-twisted RS codes were used for obtaining LCD MDS codes. New non-GRS LCD MDS codes based on twisted RS codes were also presented in \cite{wu2021new}. In \cite{huang2020mds}, self-dual MDS and near MDS codes were constructed using twisted RS codes for $t=1$ and $h=k-1$. For this twist and hook, a parity-check matrix was given as well.
An AG variant of twisted RS codes was investigated in \cite{math9010040} using codes coming from the Hermitian curve.
Further results on twisted Reed--Solomon codes can also be found in the dissertation of the second author~\cite{puchinger2018construction}.
In \cite{puchinger2017further}, the construction of twisted RS codes with multiple twists was used to further generalize the class of twisted Gabidulin codes \cite{sheekey2015new} in the rank metric.
Recently, a twisted variant of linearized Reed--Solomon codes (a mix of Reed--Solomon and Gabidulin codes, considered in the sum-rank metric) was proposed in \cite{neri2021twisted}.

\section{Preliminaries}

For $\alphaVec = [\alpha_1,\dots,\alpha_n] \in \Fq^n$, we define the evaluation map
\begin{align*}
\ev{\cdot}{\alphaVec} \, : \, \Fq[X] &\to \Fq^n, \\
f &\mapsto \left[ f(\alpha_1),\dots,f(\alpha_n) \right].
\end{align*}

Note that $\ev{\cdot}{\alphaVec}$ is an $\Fq$-linear map.
If the $\alpha_i$ are distinct, then the restriction of $\ev{\cdot}{\alphaVec}$ to polynomials of degree $<n$, i.e.,  $\ev{\cdot}{\alphaVec}\big\rvert_{\Fq[X]_{<n}}$, is bijective.

For distinct evaluation points $\alpha_1,\dots,\alpha_n \in \Fq$ and arbitrary column multipliers $v_1,\dots,v_n \in \Fq^\ast$, the corresponding \emph{generalized Reed--Solomon (GRS) code} of dimension $k$ is defined by
\begin{equation*}
\Code_\mathsf{GRS} = \ev{\Fq[X]_{<k}}{\alphaVec}\cdot \diag(v_1,\dots,v_n).
\end{equation*}
Two linear codes are called \emph{(monomially) equivalent} if one code can be obtained from the other by permutation of codeword positions and entry-wise multiplication with non-zero field elements.
In particular, any GRS code is equivalent to an RS code.

\section{Twisted Reed--Solomon Codes}\label{sec:twisted_rs_codes}

In this section, we define twisted Reed--Solomon codes and show some of their properties.

\begin{definition}\label{def:twisted_RS_codes}
Let $n,k,\numTwists \in \NN$ be positive integers with $k < n$.
We call $\ell$ the \emph{number of twists}.
Futhermore, choose three vectors
\begin{itemize}
\item $\tVec = [t_1,\dots,t_\ell] \in \{1,\dots,n-k\}^\numTwists$ (called \emph{twist vector}),
\item $\hVec = [h_1,\dots,h_\ell] \in \{0,\dots,k-1\}^\numTwists$ (called \emph{hook vector}),
\item $\etaVec = [\eta_1,\dots,\eta_\ell] \in \Fq^\numTwists$ (called \emph{coefficient vector})
\end{itemize}
such that the tuples $[h_i,t_i]$ for $i=1,\dots,\ell$ are distinct.\footnote{This means that the $\tVec$ and $\hVec$ vectors may have repeated entries, just not in the same coordinates.}

We define the set of \emph{\twisted polynomials} by
\begin{equation*}
\evpolys = \left\{ f = \sum_{i=0}^{k-1} f_i X^i + \sum_{j=1}^{\numTwists} \eta_j f_{h_j} X^{k-1+t_j} : f_i \in \Fq \right\}.
\end{equation*}
Let $\alpha_1,\dots,\alpha_n \in \Fq$ be distinct and write $\alphaVec = [\alpha_1,\dots,\alpha_n]$.
The corresponding \emph{\twistedC Reed--Solomon code} is defined by
\begin{equation*}
\Cmult := \ev{\evpolys}{\alphaVec} \subseteq \Fq^n.
\end{equation*}
For brevity, we often say \emph{twisted polynomials} and \emph{twisted RS codes}, respectively.
\end{definition}

\begin{lemma}\label{lem:evpolys_basis}
The set of \twisted polynomials $\evpolys$ is a $k$-dimensional subspace of $\Fq[X]$. A basis of $\evpolys$ is given by
\begin{align}
g_i := X^i + \sum_{\substack{j=1 \\ h_j=i}}^{\ell} \eta_j X^{k-1+t_j} \label{eq:evpolys_canonical_basis}
\end{align}
for $i=0,\dots,k-1$.
\end{lemma}

\begin{proof}
For any $f = \sum_{i=0}^{k-1} f_i X^i + \sum_{j=1}^{\numTwists}  \eta_j f_{h_j} X^{k-1+t_j}\in \evpolys$, we can write $f = \sum_{i=0}^{k-1} f_i g_i$, where $f_i \in \Fq$.
Furthermore, $g_i \in \evpolys$ and the $g_i$ are linearly independent since the monomial $X^i$ appears in $g_i$ only for each $i=0,\dots,k-1$ (note that $k-1+t_j > k-1$).
\end{proof}

\begin{proposition}\label{prop:twisted_codes_linear}
A \twistedC Reed--Solomon code $\Cmult$ is a linear $[n,k]$ code. With $g_0,\dots,g_{k-1} \in \Fq[X]$ as in \eqref{eq:evpolys_canonical_basis}, the matrix
\begin{align}
\Gmat := \begin{bmatrix}
\ev{g_0}{\alphaVec} \\
\vdots \\
\ev{g_{k-1}}{\alphaVec} 
\end{bmatrix} \in \Fq^{k \times n} \label{eq:generator_matrix_canonical}
\end{align}
is a generator matrix of $\Cmult$.
\end{proposition}

\begin{proof}
Since $\ev{\cdot}{\alphaVec}$ is $\Fq$-linear and $\evpolys$ is an $\Fq$-vector space, the code $\Cmult$ is linear. Furthermore, we have $\deg f < n$ for all $f \in \evpolys$ due to $t_i \leq n-k$.
Hence, $\ev{\cdot}{\alphaVec}$ is injective on the evaluation polynomials, which implies
$\dim_{\Fq}\!\left(\Cmult\right) = \dim_{\Fq}\!\left(\evpolys\right) = k$.
The same argument implies that the $\ev{g_i}{\alphaVec}$ are a basis of $\Cmult$, i.e., $\Gmat$ is in fact a generator matrix.
\end{proof}

\begin{remark}
Some remarks about Definition~\ref{def:twisted_RS_codes}.
\begin{itemize}
\item Twisted RS codes are not related to twisted BCH codes as defined in \cite{edel1997twisted}. The name is inspired by Sheekey's twisted Gabidulin codes \cite{sheekey2015new}, which are related to \emph{generalized twisted fields}.
\item The condition that the tuples $[h_i,t_i]$ are distinct is no restriction in general. Assume that $[h_i,t_i] = [h_j,t_j]$ for some $i \neq j$. Then we obtain the same code by removing $h_j,t_j,\eta_j$ from the twist, hook and coefficient vector (note that the number of twists decreases), respectively, and replacing $\eta_i$ by $\eta_i+\eta_j$. We can repeat this process until all tuples are distinct.
\item Setting $\eta_i \neq 0$ for all $i$ is in principle no restriction if we are interested in codes that are not obviously RS codes. However, we allow the $\eta_i$ to be $0$ such that the family of twisted codes includes RS codes in a natural way.
\item The restriction $t_i \leq n-k$ is not necessary for $\ev{\cdot}{\alphaVec}$ to be injective on $\evpolys$. Hence, it might be possible to relax this condition.
A necessary and sufficient condition is that the polynomials 
\begin{equation*}
g_i \mod \prod_{j=1}^{n}(X-\alpha_j)
\end{equation*}
for $i=0,\dots,k-1$ with $g_i$ as in \eqref{eq:evpolys_canonical_basis} are linearly independent.
This condition is obviously fulfilled if $\deg g_i <n$ since $\deg \prod_{j=1}^{n}(X-\alpha_j) = n$, but it would require rather technical conditions on $\alphaVec$, $\tVec$, $\hVec$, and $\etaVec$ to guarantee it for for $\deg g_i \geq n$.
\end{itemize}
\end{remark}

\begin{example}
We give three example generator matrices.
For easier notation, we write $\alphaVec^i := [\alpha_1^i,\alpha_2^i, \dots, \alpha_n^i]$.

For $\etaVec = \0$, we obtain a Reed--Solomon code since the basis of $\evpolys$ given in Lemma~\ref{lem:evpolys_basis} is $g_i = X^i$ and, hence, the generator matrix in Proposition~\ref{prop:twisted_codes_linear} is a Vandermonde matrix
\begin{align*}
\Gmat = \begin{bmatrix}
\alphaVec^0 \\
\alphaVec^1 \\
\vdots \\
\alphaVec^{k-1}
\end{bmatrix}
=\begin{bmatrix}
\alpha_1^0 & \dots & \alpha_n^0 \\
\alpha_1^1 & \dots & \alpha_n^1 \\
\vdots & \ddots & \vdots \\
\alpha_1^{k-1} & \dots & \alpha_n^{k-1} \\
\end{bmatrix}.
\end{align*}
For $q=n=9$, $k=5$, $\ell=1$, $h_1=2$, $t_1=2$, and $\eta_1$ a non-square of $\mathbb{F}_9$, we obtain a punctured Glynn's code \cite{glynn1986non} (by evaluating ``at infinity'' (cf.~Remark~\ref{rem:evaluation_at_infinity}) in addition, we get exactly Glynn's code).
Glynn's code is the first-known MDS code with odd field size, length $n=q+1$, and dimension $3 \leq k \leq q-1$ that is not a Generalized Reed--Solomon code.
The generator matrix in Proposition~\ref{prop:twisted_codes_linear} is given by
\begin{align*}
\Gmat = \begin{bmatrix}
\alphaVec^0 \\
\alphaVec^1 \\
\alphaVec^2+\eta_1\alphaVec^6 \\
\alphaVec^3 \\
\alphaVec^4 \\
\end{bmatrix}
\end{align*}

For $q \geq n>7$, $k = 5$, $\tVec = [1,3,3]$, and $\hVec = [4,4,2]$, the generator matrix in Proposition~\ref{prop:twisted_codes_linear} is of the form
\begin{align*}
\Gmat = \begin{bmatrix}
\alphaVec^0 \\
\alphaVec^1 \\
\alphaVec^2 + \eta_3 \alphaVec^7 \\
\alphaVec^3 \\
\alphaVec^4 + \eta_1 \alphaVec^5  + \eta_2 \alphaVec^7
\end{bmatrix}.
\end{align*}
\end{example}

\section{MDS Twisted RS Codes}
\label{sec:MDS}

Not all twisted RS codes are MDS. In this section, we give several families of MDS twisted RS codes.

\subsection{A General MDS Condition}\label{ssec:MDS_general}

\begin{definition}\label{def:sum-product-free}
Let $\Fq/\Fqsmall$ be a field extension.
A vector $\etaVec \in \Fq^\ell$ is called \emph{$\Fqsmall$-sum-product free} if 
\begin{equation*}
\sum_{\substack{\mathcal{S} \subseteq \{1,\dots,\ell\} \\ \mathcal{S} \neq \emptyset}} a_{\mathcal{S}} \prod_{i \in \mathcal{S}} \eta_i \notin \Fqsmall^\ast \quad \forall \, a_{\mathcal{S}} \in \Fqsmall.
\end{equation*}
\end{definition}

Equivalently, a vector $\etaVec$ is $\Fqsmall$-sum-product free exactly when there is no polynomial $f \in \Fqsmall[X_1,\ldots,X_\ell]$ with non-zero constant coefficient and of degree at most $1$ in each $X_i$ such that $f(\eta_1,\ldots,\eta_\ell) = 0$.

\begin{proposition}
Let $\Fq / \Fqsmall$ be an extension of finite fields.
Let $k < n \leq q_0$ and let $\alpha_1,\ldots,\alpha_n \in \Fqsmall$ be distinct.
For any $\tVec, \hVec$ and $\vec \eta$ chosen as in \cref{def:twisted_RS_codes} and such that $\etaVec$ is $\Fqsmall$-sum-product free, then the twisted RS codes $\Cmult$ is MDS.
\end{proposition}

\begin{proof}
Let $\Gmat$ be the generator matrix of $\Cmult$ given in \eqref{eq:generator_matrix_canonical}.
Since each $\alpha_i \in \Fqsmall$, we can consider the entries of $\Gmat$ to be in $\Fqsmall[\eta_1,\ldots,\eta_\ell]$.
The code $\Cmult$ is MDS if and only if every $k \times k$ minor of $\Gmat$ is non-zero.
Observe that each $\eta_i$ appears in exactly one row of $\Gmat$, and hence as a polynomial in $\eta_1,\ldots,\eta_\ell$, any such $k \times k$ minor has degree at most one in each variable.
Moreover, its constant term is non-zero since setting $\eta_1 = \ldots = \eta_\ell = 0$ yields an RS code which is MDS, and hence has non-zero $k \times k$ minors.
Since $\etaVec$ is $\Fqsmall$-sum-product free, then every such expression is non-zero.
\end{proof}

The following gives two constructions of $\Fqsmall$-sum-product free sets:

\begin{proposition}\label{prop:MDS_condition_subfield_chain}
  Let $\Fqsmall \subsetneq \F_{q_1} \subsetneq \ldots \subsetneq \F_{q_\ell} = \Fq$ be a proper chain of subfields.
  Let $\eta_1, \ldots, \eta_\ell$ be chosen with the condition that $\eta_i \in \F_{q_i} \setminus \F_{q_{i-1}}$.
  Then $\etaVec := [\eta_1, \ldots, \eta_\ell]$ is $\Fqsmall$-sum-product free.
\end{proposition}

\begin{proof}
We prove the claim by induction on $\ell$.
If $\ell=1$, we have $a \eta_1 \notin \Fqsmall^\ast$ for any $a \in \Fqsmall$ since $\eta_1 \notin \Fqsmall$.
For the inductive step, we can split any sum product
\begin{equation*}
	A := \sum_{\substack{\mathcal{S} \subseteq \{1,\dots,\ell\} \\ \mathcal{S} \neq \emptyset}} a_{\mathcal{S}} \prod_{i \in \mathcal{S}} \eta_i = a(\eta_1,\dots,\eta_{\ell-1}) + \eta_\ell b(\eta_1,\dots,\eta_{\ell-1}),
\end{equation*}
where $a,b \in \Fqsmall[X_1,\dots,X_{\ell-1}]$ are polynomials with degree at most one in each variable $X_i$ and $a$ has zero constant term (i.e., $a(\eta_1,\dots,\eta_{\ell-1})$ is a sum-product of $\eta_1,\dots,\eta_{\ell-1}$).
	By the inductive step and the choice of the $\eta_i$, we have $a(\eta_1,\dots,\eta_{\ell-1}) \in \F_{q_{\ell-1}} \setminus \Fqsmall^\ast$ and $b(\eta_1,\dots,\eta_{\ell-1}) \in \F_{q_{\ell-1}}$.
If $b(\eta_1,\dots,\eta_{\ell-1})=0$, then we have $A = a(\eta_1,\dots,\eta_{\ell-1}) \notin \Fqsmall^\ast$.
Else, we have $A \notin \F_{q_{\ell-1}}$ since otherwise, $\eta_\ell = \tfrac{A-a(\eta_1,\dots,\eta_{\ell-1})}{b(\eta_1,\dots,\eta_{\ell-1})}$ would be in $\F_{q_{\ell-1}}$. In particular, we have $A \notin \Fqsmall$.
\end{proof}

\begin{proposition}\label{prop:MDS_condition_power_basis}
Let $\F_q / \Fqsmall$ be an extension of finite fields of degree at least $\ell+1 \geq 2$, and let $1, \psi, \ldots, \psi^{[F_q : \Fqsmall] - 1} \in \F_q$ be a power basis of the extension.
Then any $\etaVec := [a_1 \psi, \ldots, a_\ell \psi]$ with $a_i \in \Fqsmall \setminus \{0\}$ is $\Fqsmall$-sum-product free.
\end{proposition}

\begin{proof}
For any non-empty $\mathcal{I} \subseteq \{1,\dots,\ell\}$, we have $\prod_{i \in \mathcal{I}} \eta_i = b \psi^{|\mathcal{I}|}$ for some $b \in \Fqsmall \setminus \{0\}$.
Hence, an $\Fqsmall$-linear combination of such terms must be of the form $b_1 \psi + b_2 \psi^2 + \ldots + b_\ell \psi^\ell$ with $b_i \in \Fqsmall$, and this will never be 0 as $\Fq$ has degree at least $\ell+1$ over $\Fqsmall$.
\end{proof}

\begin{remark}
The propositions above provide two constructions of MDS twisted RS codes.
Both methods require that the evaluation points are chosen from a proper subfield of the code's base field.
Since $n \leq q_0$ and the smallest prime number $q_0$ greater or equal to $n$ satisfies $q_0 < 2n$ by Bertrand's postulate,
the smallest overall field sizes $q$ for the constructions fulfill
\begin{align*}
n^{2^{\ell}} &\leq q = q_0^{2^{\ell}} < (2n)^{2^{\ell}},  &&\text{(Proposition~\ref{prop:MDS_condition_subfield_chain}),} \\
n^{\ell+1} &\leq q = q_0^{\ell+1} < (2n)^{\ell+1}, &&\text{(Proposition~\ref{prop:MDS_condition_power_basis}).}
\end{align*}
For these smallest-possible field sizes, we have
\begin{align*}
\tfrac{1}{2} q^{2^{-\ell}} &< n \leq q^{2^{-\ell}} &&\text{(Proposition~\ref{prop:MDS_condition_subfield_chain}),}\\
\tfrac{1}{2} q^{\frac{1}{\ell+1}} &< n \leq q^{\frac{1}{\ell+1}}
&&\text{(Proposition~\ref{prop:MDS_condition_power_basis}).}
\end{align*}
It can be seen that, compared to Proposition~\ref{prop:MDS_condition_subfield_chain}, the construction in Proposition~\ref{prop:MDS_condition_power_basis} is able to provide smaller field sizes for any given code length. We include Proposition~\ref{prop:MDS_condition_subfield_chain} for the sake of having a second construction that might prove useful for other purposes than merely minimizing the field size. For instance, there is an analog of Proposition~\ref{prop:MDS_condition_subfield_chain} (based on a conference version of this paper) for constructing twisted Gabidulin codes in the rank metric \cite{puchinger2017further}, but there is no rank-metric analog of Proposition~\ref{prop:MDS_condition_power_basis}, yet, and it is not obvious how to adapt it.
\end{remark}

\subsection{$(\ast)$-Twisted RS Codes}\label{ssec:startwisted_codes}

The $\Fqsmall$-sum-product free property as in the previous subsection is a rather strong restriction on the $\eta_i$ and yields relatively short MDS codes. %
In this and the following subsection, we will obtain longer MDS codes for two specific choices of $\tVec$ and $\hVec$.

We first consider twisted RS codes with one twist $\ell=1$, hook $\hVec=0$, and twist $\tVec = 1$.
In this case, we have the following necessary and sufficient MDS condition on the $\alpha_i$ and $\eta_1$.

\begin{lemma}\label{lem:startwisted_MDS}
Let $\ell=1$ and $n,k,\alphaVec,\etaVec$ be chosen as in \cref{def:twisted_RS_codes}.
The code $\TRS{\alphaVec,1,0,\etaVec}{n,k}$ is MDS if and only if
\begin{equation}
\eta_1 (-1)^k \prod_{i \in \Iset} \alpha_i \neq 1 \quad \forall \, \Iset \subseteq \{1,\dots,n\} \text{ s.t. } |\Iset|=k. \label{eq:startwisted_MDS_condition}
\end{equation}
\end{lemma}

\begin{proof}
All evaluation polynomials are of the form $f = \sum_{i=0}^{k-1} f_i x^i + \eta_1 f_0 x^k$.
If $f_0=0$, the weight of the corresponding codeword is either $0$ or at least $n-k+1$ since $\deg f < k$.
Otherwise, $f$ corresponds to a codeword of weight $<n-k+1$ if and only if $f$ has exactly $k$ roots among the $\alpha_i$, i.e., there is a subset $\Iset \subseteq \{1,\dots,n\}$ with $|\Iset|=k$ such that $f = \eta_1 f_0 \prod_{i \in \Iset} (x-\alpha_i)$.
The constant term of $f$ is $f_0 = f(0) = \eta_1 f_0 \prod_{i \in \Iset} (-\alpha_i)$. Due to $f_0\neq0$, we have
\begin{equation*}
\eta_1 (-1)^k \prod_{i \in \Iset} \alpha_i = 1.
\end{equation*}
Hence, all non-zero evaluation polynomials have at most $k-1$ roots among the $\alpha_i$ if and only if \eqref{eq:startwisted_MDS_condition} is satisfied.
\end{proof}

For $\eta_1 \neq 0$, a sufficient condition for \eqref{eq:startwisted_MDS_condition} to be fulfilled is that $(-1)^k \eta_1^{-1}$ is not contained in the multiplicative group generated by the $\alpha_i$'s.
This motivates the following definition.

\begin{definition}\label{def:startwisted}
Let $G$ be a proper subgroup of $(\F_q^*, \cdot)$, $\alpha_i \in G \cup \{0\}$ for all $i$, and $(-1)^k \eta_1^\mo \in \F_q^* \setminus G$.
Then, we call $\TRS{\alphaVec,1,0,\etaVec}{n,k}{}$ a \startw code.
\end{definition}

\begin{theorem}\label{thm:startwisted_MDS}
Any \startw code is MDS.
\end{theorem}

\begin{proof}
For any $\Iset \subseteq \{1,\dots,n\}$, we have $\prod_{i \in \Iset} \alpha_i \in G \cup \{0\}$.
Since $(-1)^k\eta_1^\mo \notin G \cup \{0\}$, Condition \eqref{eq:startwisted_MDS_condition} is fulfilled, and the code is MDS by Lemma~\ref{lem:startwisted_MDS}.
\end{proof}

For any divisor $a>1$ of $q-1$, there is a proper subgroup $G$ of $\Fq^*$ of cardinality $(q-1)/a$.
This means that \startw codes can have length $n = \tfrac{q+1}{a}$ and can be rather long compared to the constructions in the previous subsection.
In particular, if $q$ is odd, then $2 \mid q-1$ and $n = \tfrac{q+1}{2}$ is possible.

For even $q$, there is no multiplicative subgroup of this cardinality.
However, if we allow arbitrary evaluation points and $\eta \in \Fq^*$, MDS twisted RS codes  with $\tVec = 1$, $\hVec=0$, and length $n \approx q/2$ may exist for even $q$: our computer search (cf.~\cite{beelen2017twisted} %
shows, e.g., for $q=16$, there are many such codes of length $n = 9$ for $k = 3,4,5$.

Choosing the evaluation points from a multiplicative group appears to be rather restrictive.
However, the following analysis shows that for odd $q$, \startw codes have maximal length among all MDS twisted RS codes with $\tVec=1$ and $\hVec=0$.
To show this, we use the notion of $k$-sum generators in finite abelian groups, which was introduced in \cite{roth_construction_1989, roth_t-sum_1992} and originally used to construct non-RS MDS codes.

\begin{definition}[\!\!\cite{roth_t-sum_1992}]\label{def:k-sum_generator}
Let $(A,\oplus)$ be a finite abelian group and $k \in \NN$.
A \emph{$k$-sum generator} of $A$ is a subset $S \subseteq A$ such that for any $a \in A$, there are distinct $s_1,\dots,s_k \in S$ with $a=\bigoplus_{i=1}^{k} s_i$.
The smallest integer such that any $S \subset A$ with $|S|>M(k,A)$ is a $k$-sum generator of $A$ is denoted by $M(k,A)$.
\end{definition}

\begin{lemma}[{\!\!\cite{roth_t-sum_1992}}]\label{lem:tsum_roth_theorem_even_order}%
Let $A$ be a finite abelian group of order $|A|= 2r$ for some $r \geq 6$. For any $k$ with $3 \leq k \leq r-2$, we have
\begin{equation*}
M(k,A) = \begin{cases}
r+1, &\parbox[t]{5cm}{if $A \in \{\ZZ_2^m,\ZZ_4 \times \ZZ_2^{m-1}\}$ for some $m > 1$ and $k \in \{3,r-2\}$,} \\
r, &\text{else}.
\end{cases}
\end{equation*}
\end{lemma}

\begin{lemma}\label{lem:trs_k_sum_gen_star_MDS}
Let $k,n,\alphaVec,\etaVec$ be chosen as in \cref{def:twisted_RS_codes} such that $S := \{\alpha_1,\dots,\alpha_n\} \subseteq \Fq^*$ is a $k$-sum generator of $(\Fq^*,\cdot)$ and $\etaVec \in \Fq^*$.
Then, the code $\TRS{\alphaVec,1,0,\etaVec}{n,k}{}$ is not MDS.
\end{lemma}

\begin{proof}
Since $S$ is a $k$-sum generator of $(\Fq^*,\cdot)$ and $(-1)^k \eta_1^{-1} \neq 0$, there is an index set $\Iset \subseteq \{1,\dots,n\}$ with $|\Iset| = k$ such that $\prod_{i \in \Iset} \alpha_i = (-1)^k \eta_1^{-1}$. 
Lemma~\ref{lem:startwisted_MDS} then implies the claim.
\end{proof}

\begin{theorem}\label{thm:trs_inverse_mds_statement_very_simple_twist}
Let $q$ be odd and $3 \leq k \leq \tfrac{q-1}{2}-2$.
If $n > \frac{q+1}{2}$, then $\TRS{\alphaVec,1,0,\etaVec}{n,k}{}$ is not MDS for any choice of $\alphaVec$ and $\etaVec \neq 0$ as in \cref{def:twisted_RS_codes}.
\end{theorem}

\begin{proof}
The abelian group $(\Fq^*,\cdot)$ is cyclic and of even order $|\Fq^*|=q-1$ since $q$ is odd. Thus, \cref{lem:tsum_roth_theorem_even_order} implies
that the maximal cardinality of a subset of $\Fq^*$ that is not a $k$-sum generator is
$M(k,\Fq^*) = \tfrac{q-1}{2}$.
Since, for $S := \{\alpha_1,\dots,\alpha_n\}$, we have $|S \setminus \{0\}| \geq n-1 >M(k,\Fq^*)$, the set $S$ is therefore a $k$-sum generator of $\Fq^*$.
By \cref{lem:trs_k_sum_gen_star_MDS}, the code $\TRS{\alphaVec,1,0,\eta}{n,k}{}$ is not MDS.
\end{proof}

\subsection{$(+)$-Twisted Reed--Solomon Codes}\label{ssec:trs_plus_twisted_codes}

We consider twisted RS codes with one twist $\ell=1$, hook $\hVec = k-1$, and twist $\tVec = 1$. In this case, we can also give a necessary and sufficient MDS condition, which can be seen as the additive analog of Lemma~\ref{lem:startwisted_MDS}. It gives rise to a similar construction as the \startw codes using additive instead of multiplicative subgroups of $\Fq$.

\begin{lemma}\label{lem:trs_plus_twisted_MDS}
Let $\ell=1$ and $n,k,\alphaVec,\etaVec$ be chosen as in \cref{def:twisted_RS_codes}.
The code $\TRS{\alphaVec,1,k-1,\etaVec}{n,k}$ is MDS if and only if
\begin{equation}
\eta \sum_{i \in \Iset} \alpha_i \neq -1 \quad \forall \, \Iset \subseteq \{1,\dots,n\} \textrm{ s.t. } |\Iset|=k. \label{eq:trs_cond_MDS_primitively_twisted_h=k-1}
\end{equation}
\end{lemma}

\begin{proof}
The code is MDS if and only if any non-zero evaluation polynomial has at most $k-1$ zeros among the $\alpha_i$.

Suppose that there is a polynomial $f \in \evpolys \setminus \{0\}$ with $k$ roots among the $\alpha_i$.
Then, we have $f_{k-1} \neq 0$ and there is a set $\Iset \subseteq \{1,\dots,n\}$ with $|\Iset|=k$ and $f = f_k \prod_{i \in \Iset} (x-\alpha_i)$, i.e., $f_{k-1} = f_k \sum_{i \in \Iset} (-\alpha_i)$.
Due to the choice of $\tVec$ and $\hVec$, we have $f_{k} = \eta f_{k-1}$, so
$\eta \sum_{i \in \Iset} \alpha_i = -1$
for this $\Iset$.

Conversely, assume that there is such a set $\Iset$ with $\eta \sum_{i \in \Iset} \alpha_i = -1$.
Then, $f = \eta \prod_{i \in \Iset} (x-\alpha_i)$ is an evaluation polynomial and has $k$ roots among the evaluation points.
\end{proof}

Analog to the multiplicative case, a sufficient condition for \eqref{eq:trs_cond_MDS_primitively_twisted_h=k-1} to be fulfilled is to choose the evaluation points from a proper subgroup of $(\Fq,+)$ and $-\eta^{-1}$ not in this subgroup.
This gives the following class of MDS twisted RS codes.

\begin{definition}\label{def:trs_plustwisted}
Let $V$ be a proper subgroup of $(\Fq,+)$, $\etaVec^\mo \in \Fq \setminus V$, and $\alphaVec$ consist of $n$ distinct elements of $V$.
Then, $\TRS{\alphaVec,1,k-1,\etaVec}{n,k}{}$ is called a \emph{\plustw code}. %
\end{definition}

\begin{theorem}\label{thm:MDS_property_primitively_twisted_h=k-1}
Any \plustw code is MDS.
\end{theorem}

\begin{proof}
This follows immediately from \cref{lem:trs_plus_twisted_MDS}.
\end{proof}

If $p$ is the characteristic of $\Fq$, then there is a proper subgroup $V$ of $(\Fq,+)$ with order $q/p$.
Hence, a \plustw code can have length up to $n = \tfrac{q}{p}$. In particular, for even $q$, we get codes of length $n = \tfrac{q}{2}$.

\begin{remark}\label{rem:evaluation_at_infinity}
For $\ell=1$ and general $h_1$ and $t_1$, we define the evaluation at infinity as $f(\infty) := f_{k-1+t_1}$ (note that $k-1+t_1$ is the maximal degree of a polynomial in $\evpolys$). Due to $(\alpha f+\beta g)(\infty) = \alpha f(\infty)+ \beta g(\infty)$ for all $f,g \in \evpolys$ and $\alpha,\beta \in \Fq$, adding $\infty$ to the evaluation point set gives a linear code.
For $h_1=k-1$, we have $f(\infty)=0$ if and only if $\deg(f)<k-1$.
Hence, if a twisted RS code with these parameters and $\alphaVec \in \Fq^n$ is MDS, then the ``extended'' code with the same $k,\tVec,\hVec,\etaVec$ and evaluation points $\alphaVec' := [\alpha_1,\dots,\alpha_n,\infty]$ is also MDS.
By extending a \plustw code, we get an MDS code of length up to $n = \tfrac{q}{2}+1$ over a field of characteristic $2$.
\end{remark}

As in the \startw case, we study the maximal length of a twisted RS code with $\tVec=1$ and $\hVec=k-1$, and arbitrary $\alphaVec$ and $\etaVec \neq 0$. The proofs of the following two statements are similar to those of \cref{lem:trs_k_sum_gen_star_MDS} and \cref{thm:trs_inverse_mds_statement_very_simple_twist}, respectively, and are therefore omitted.

\begin{lemma}\label{lem:trs_k_sum_gen_plus_MDS}
Let $k,n,\alphaVec,\etaVec$ be chosen as in \cref{def:twisted_RS_codes} such that $S := \{\alpha_1,\dots,\alpha_n\} \in \Fq$ is a $k$-sum generator of $(\Fq,+)$.
Then, the code $\TRS{\alphaVec,1,k-1,\etaVec}{n,k}{}$ is not MDS.
\end{lemma}

\begin{theorem}\label{thm:trs_inverse_mds_statement_very_simple_twist_h=k-1}
Let $q$ be even and $3 \leq k \leq \tfrac{q}{2}-2$.
If the code length satisfies
\begin{equation*}
n > \begin{cases}
\frac{q}{2}, &\text{if } 3 < k < \frac{q}{2}-3, \\
\frac{q}{2}+1, &\text{if } k \in \{3,\frac{q}{2}-2\},
\end{cases}
\end{equation*}
then the twisted code $\TRS{\alphaVec,1,k-1,\eta}{n,k}{}$ is not MDS for any choice of $\etaVec$ as in \cref{def:twisted_RS_codes}.
\end{theorem}

\section{Duals of Twisted RS Codes}\label{sec:duals}

In this section, we show that the family of twisted RS codes whose evaluation points form a multiplicative group are closed under duality.
We use the following auxiliary matrices.

\begin{definition}\label{def:trs_duals_aux_mat}
Let $r \in \ZZ_{>0}$ and $\alphaVec \in \Fq^r$.
\begin{enumerate}[label=\roman*)]
\item The \emph{reversal matrix} $\J_r \in \Fq^{r \times r}$ is the square matrix
\begin{equation*}
\J_r = \begin{bmatrix}
 & & 1 \\
 & \iddots & \\
1 & &
\end{bmatrix}.
\end{equation*}
\item The \emph{Vandermonde matrix of $\alphaVec$} is denoted by
\begin{equation*}
\Van_r(\alphaVec) =
	\begin{bmatrix}
		\alpha_1^0 & \alpha_2^0 & \ldots & \alpha_r^0 \\
		\alpha_1^1 & \alpha_2^1 & \ldots & \alpha_r^1 \\
		\vdots    &  & \ddots & \vdots \\
		\alpha_1^{r-1} & \alpha_2^{r-1} & \ldots & \alpha_r^{r-1} \\
	\end{bmatrix}.
\end{equation*}
\end{enumerate}
\end{definition}

For a matrix $\A \in \Fq^{r \times r'}$, then $\J_r \A$ is $\A$ with the rows in reverse order.
Similarly, $\A \J_{r'}$ is $\A$ with the columns in reverse order.
And $\B := \J_r \A \J_{r'} \in \Fq^{r \times r'}$ is $\A$ ``rotated'', i.e., $B_{i,j} = A_{r-i+1, r'-j+1}$.
If the $\alpha_i$ form a multiplicative group, we can give the inverse of the Vandermonde matrix $\Van_n(\alphaVec)$ with the help of the reversal matrix as follows. This is a reformulation of a result in \cite{althaus1969inverse}.

\begin{lemma}\label{lem:trs_duals_Vandermonde_inverse}
Let $\alphaVec \in \Fq^n$ such that the $\alpha_i$ are distinct and form a multiplicative subgroup of $\Fq^*$. Then,
\begin{equation*}
\left( \Van_n(\alphaVec)^\top \right)^{-1} = \J_n \cdot \Van_n(\alphaVec) \cdot \diag(\alphaVec/n).
\end{equation*}
\end{lemma}

\begin{proof}
Since the entries of $\alphaVec$ are a multiplicative group, we have $\prod_{i=1}^{n} (x-\alpha_i) = x^n - 1$ and
\begin{align*}
\left( \Van_n(\alphaVec)^\top \right)^{-1}
&= \frac{1}{n}
 \begin{bmatrix}
   1                 & 1                 & \dots  & 1             \\
   \alpha_1^{-1}     & \alpha_2^{-1}     & \dots  & \alpha_n^{-1} \\
   \vdots            & \vdots            & \ddots & \vdots        \\
   \alpha_1^{-(n-1)} & \alpha_2^{-(n-1)} & \dots  & \alpha_n^{-(n-1)}
 \end{bmatrix} \\
&= \J_n \cdot \Van_n(\alphaVec) \cdot \diag(\alphaVec/n),
\end{align*}
where the first equality follows by \cite{althaus1969inverse}.
\end{proof}

Lemma~\ref{lem:trs_duals_Vandermonde_inverse} enables us to describe the duals of the following class of codes, which includes the family of twisted RS codes with evaluation points forming a multiplicative group.

\begin{lemma}\label{lem:trs_duals_L_twisted_dual}
Let $\Code[n,k]$ be a linear code with a generator matrix of the form
\begin{equation*}
\Gmat = [\I \mid \Lmat] \cdot \Van_n(\alphaVec),
\end{equation*}
where $\I \in \Fq^{k \times k}$ is the identity matrix, $\Lmat \in \Fq^{k \times {n-k}}$, and the entries of $\alphaVec \in \Fq^{n}$ are distinct and form a multiplicative subgroup of $\Fq^*$. Then, the following is a generator matrix of the dual code $\Code^\dual$:
\begin{equation*}
\Hmat = [\I \mid -\J_{n-k} \Lmat^\top \J_k ] \cdot \Van_n(\alphaVec) \cdot \diag(\alphaVec/n).
\end{equation*}
\end{lemma}

\begin{proof}
By construction, $\Hmat$ has full rank $n-k$ and fulfills
\begin{align*}
&\Gmat \cdot \Hmat^\top \\
&= [\I \mid \Lmat] \Van_n(\alphaVec) \cdot \Big( [\I \mid -\J_{n-k} \Lmat^\top \J_k] \Van_n(\alphaVec) \diag(\alphaVec/n) \Big)^\top \\
&= [\I \mid \Lmat] \Van_n(\alphaVec) \cdot \Big( \J_{n-k} [-\Lmat^\top \mid \I] \underset{= \,  (\Van_n(\alphaVec)^{-1})^\top \text{ (\cref{lem:trs_duals_Vandermonde_inverse})}}{\underbrace{\J_n \Van_n(\alphaVec) \diag(\alphaVec/n)}} \Big)^\top \\[-0.4cm]
&= [\I \mid \Lmat] \begin{bmatrix}
-\Lmat \\ \I
\end{bmatrix} \J_{n-k} = \ve 0
\end{align*}
so it is a parity-check matrix of $\Code$, and thus, a generator matrix of the dual code.
\end{proof}

\cref{lem:trs_duals_L_twisted_dual} implies the following duality statement for twisted RS codes with evaluation points forming a multiplicative group.

\begin{theorem}\label{thm:trs_duals}
Let $n,k,\alphaVec,\tVec,\hVec,\etaVec$ be chosen as in \cref{def:twisted_RS_codes} such that the entries of $\alphaVec$ form a multiplicative subgroup of $\Fq^*$.
Then, the dual of $\Cmult$ is equivalent to $\TRS{\alphaVec,k-\hVec,n-k-\tVec,-\etaVec}{n,n-k}$,
where $k-\hVec := [k-h_1,\dots,k-h_\ell]$ and $n-k-\tVec$ is defined analogously. %
\end{theorem}

\begin{proof}
The canonical generator matrix (as in \eqref{eq:generator_matrix_canonical}) of any twisted RS code $\Cmult$ can be written as
\begin{equation*}
\Gmat = [\I \mid \Lmat] \cdot \Van_n(\alphaVec),
\end{equation*}
where the entries of $\Lmat \in \Fq^{k \times n-k}$ are of the form
\begin{equation*}
L_{i,j} = \begin{cases}
\eta_\mu, & \text{if } [i,j] = [h_\mu+1,t_\mu], \\
0, & \text{else}.
\end{cases}
\end{equation*}
Since we assume that the $\alpha_i$ form a multiplicative group, we can apply \cref{lem:trs_duals_L_twisted_dual} and obtain the following generator matrix of the dual code:
\begin{equation*}
\Hmat = [\I \mid -\J_{n-k} \Lmat^\top \J_k ] \cdot \Van_n(\alphaVec) \cdot \diag(\alphaVec/n)
\end{equation*}
Hence, the dual of $\Cmult$ is equivalent to a code $\Code'$ with generator matrix $[\I \mid -\J_{n-k} \Lmat^\top \J_k ] \cdot \Van_n(\alphaVec)$.
Since the entries of $\B := -\J_{n-k} \Lmat^\top \J_k$ are of the form
\begin{equation*}
B_{i,j}= \begin{cases}
-\eta_\mu, & \text{if } [i,j] = [n-k-t_\mu+1,k-h_\mu] \\
0, & \text{else},
\end{cases}
\end{equation*}
we have $\Code' = \TRS{\alphaVec,k-\hVec,n-k-\tVec,-\etaVec}{n,n-k}$, which proves the claim.
\end{proof}

A suitable example of twisted RS codes with evaluation points forming a multiplicative group are the \startw codes described in Section~\ref{ssec:startwisted_codes}.

\begin{corollary}
Let $G$ be a proper subgroup of $(\Fq^*,\cdot)$ and $\ell=1,k,n,\alphaVec,\etaVec$ be chosen as in \cref{def:twisted_RS_codes} such that $G = \{\alpha_1,\dots,\alpha_n\}$ and $(-1)^{n-k+1} \eta_1^{-1} \notin G \cup \{0\}$.
Then, the code $\TRS{\alphaVec,n-k,k-1,\etaVec}{n,k}{}$ is equivalent to the dual of the \startw code $\TRS{\alphaVec,1,0,-\etaVec}{n,n-k}{}$. In particular, it is MDS.
\end{corollary}

\begin{remark}
Theorem~\ref{thm:trs_duals} could be generalized if the inverse of $\Van_n(\alphaVec)$ could be described similarly as in Lemma~\ref{lem:trs_duals_Vandermonde_inverse} for a wider class of evaluation points $\alphaVec$.
It is not true, however, that the dual of any twisted RS code is equivalent to a twisted RS code with the same number of twists: by computer search, we found twisted RS codes over $\mathbb{F}_{11}^{}$ of length $n=8$ and with one twist ($\ell=1$) whose dual codes are not equivalent to any twisted RS code with one twist.
\end{remark}

\begin{remark}
We can generalize \cref{lem:trs_duals_L_twisted_dual} and \cref{thm:trs_duals} to allow also $\alpha_i=0$, if we in addition assume $t_i \neq n-k$ for all $i$ or $h_i \neq 0$ for all $i$.
The proof idea is as follows.
Let $\alphaVec = [\alpha_1,\dots,\alpha_n,0]$ and $\Lmat \in \Fq^{k \times n+1-k}$ be a matrix whose first row is of the form $[\l_1 \mid 0]$ with $\l_1 \in \Fq^{n-k}$.
Then,
\begin{align*}
\Hmat = &\left[\I \, \Big| \, - \begin{bmatrix}
1 & \\
\l_1^\top & \I
\end{bmatrix}
\cdot \J_{n+1-k} \Lmat^\top \J_k \right] \\ &\cdot \Van_{n+1}(\alphaVec) \cdot \diag(1/n,\dots,1/n,-1)
\end{align*}
is a valid parity-check matrix for the code with generator matrix $\Gmat = [\I \mid \Lmat] \cdot \Van_n(\alphaVec) \in \Fq^{k,n+1}$. If the first row ($t_i \neq n-k$ $\forall i$) \emph{or} the last column ($h_i \neq 0$ $\forall i$) of $\Lmat$ is zero, then we have
\begin{equation*}
- \begin{bmatrix}
1 & \\
\l_1^\top & \I
\end{bmatrix}
\cdot \J_{n+1-k} \Lmat^\top \J_k = - \J_{n+1-k} \Lmat^\top \J_k.
\end{equation*}
\end{remark}

\section{Relation to GRS Codes}

Using two different techniques, we show that many twisted RS codes are not GRS codes.
Section~\ref{ssec:inequivalence_schur_squares} uses the Schur square of a code to distinguish a large class of low-rate (and special high-rate) twisted codes from GRS codes.
In Section~\ref{ssec:inequivalence_combinartorial}, we derive a combinatorial statement, which states that if all code parameters are fixed except for $\etaVec$, either all $\etaVec$ for which the code is MDS give a GRS code, or only a few of them result in GRS codes.

\subsection{Inequivalence Based on Schur Squares}\label{ssec:inequivalence_schur_squares}

Schur squares of codes have become an increasingly studied object in coding theory in the last years due to several applications \cite{couvreur2014distinguisher,cramer_secure_2015,randriambololona2015products}.
\begin{definition}
Let $\Code[n,k]$ be a linear code. The \emph{Schur square} of $\Code$ is defined as
\begin{equation*}
\Code^2 := \left\langle \left\{ \c \star \c' \, : \, \c, \, \c' \in \Code \right\} \right\rangle,
\end{equation*}
where $\c \star \c' = [c_1 c_1', \dots, c_n c_n']$ is the Schur product of two vectors. 
\end{definition}

The dimension of the Schur product of a code is an invariant under code equivalence and satisfies
\begin{equation*}
\dim\!\left( \Code^2 \right) \leq \min\{n,\tfrac{1}{2} k (k+1) \}.
\end{equation*}
A random linear code attains this upper bound with high probability, cf.~\cite{cascudo_squares_2015}.
An MDS code has Schur square dimension at least $\dim(\Code^2) \geq \min\{n,2k-1\}$ \cite{randriambololona2015products} %
and GRS codes attain this lower bound.
We will make use of these properties in this section by showing that a large family of twisted RS codes of rate less than $1/2$ has Schur square dimension at least $2k$, and thus is non-GRS.

We start with a generic lower bound on the Schur square dimension of an evaluation code.

\begin{definition}
Let $\Pset \subseteq \Fq[x]_{<n}$ be an $\Fq$-subspace and $\alphaVec$ consist of $n$ distinct elements $\alpha_i$ of $\Fq$. We define
\begin{align*}
\Dsetsmallern{\Pset} &:= \left\{ \deg(f \cdot g) \, : \, f,g \in \Pset, \, \deg(f \cdot g)<n \right\} \text{ and} \\ %
\Dbarset{\Pset,\alphaVec} &:= \left\{ \deg(\overline{f \cdot g}) \, : \, f,g \in \Pset \right\}, 
\end{align*}
where $\overline{f} := \big(f \; \rem \; \prod_{i=1}^{n} (X-\alpha_i)\big)$ for any $f \in \Fq[X]$.
\end{definition}

\begin{lemma}\label{lem:Schur_square_evaluation_code}
Let $\alphaVec \in \Fq^n$ with distinct entries, $\Pset \subseteq \Fq[x]_{<n}$ be an $\Fq$-subspace, and $\Code = \ev{\Pset}{\alphaVec}$ be the evaluation code of $\Pset$ at the evaluation points $\alphaVec$. Then,
\begin{align*}
\Code^2 &= \ev{\left\langle fg \, : \, f,g \in \Pset \right\rangle}{\alphaVec} &&\text{and} \\
\dim \left(\Code^2\right) &\geq \left|\Dbarset{\Pset,\alphaVec}\right| \geq \left|\Dsetsmallern{\Pset}\right|.
\end{align*}
\end{lemma}

\begin{proof}
The first part of the statement follows directly from $f(\alpha) \cdot g(\alpha) = (f \cdot g)(\alpha)$ for $f,g \in \Fq[X]$ and $\alpha \in \Fq$.
Since 
\begin{align*}
\ev{\left\langle fg \, : \, f,g \in \Pset \right\rangle}{\alphaVec} = \ev{\left\langle \overline{fg} \, : \, f,g \in \Pset \right\rangle}{\alphaVec}
\end{align*}
and the evaluation $\ev{\cdot}{\alphaVec}$ is a bijection between $\Fq[X]_{<n}$ and $\Fq^n$, the Schur square dimension $\dim \left(\Code^2\right)$ is greater or equal to the dimension of $\langle\overline{f \cdot g} \, : \, f,g \in \Pset\rangle$, which in turn is lower-bounded by $\left|\Dbarset{\Pset,\alphaVec}\right|$.
Note also that $\Dsetsmallern{\Pset} \subseteq \Dbarset{\Pset,\alphaVec}$.
\end{proof}

Using Lemma~\ref{lem:Schur_square_evaluation_code}, we get the following lower bound on the Schur square dimension of twisted RS codes. %

\begin{proposition}\label{prop:trs_Schur_squares_TRS_lower_bound}
Let $\alphaVec, \tVec, \hVec$, and $\etaVec$ be as in Definition~\ref{def:twisted_RS_codes}. Denote by $g_0,\dots,g_{k-1} \in \evpolys$ the basis of $\evpolys$ given in Lemma~\ref{lem:evpolys_basis} and define $S_{\tVec,\hVec,\etaVec}^{n,k} := \left\{ \deg(g_1),\dots,\deg(g_{k-1}) \right\}$. Then,
\begin{align*}
S_{\tVec,\hVec,\etaVec}^{n,k} = &\Big(\{0,\dots,k-1\} \setminus \{h_j \, : \, \eta_j \neq 0\} \Big) \, \cup \,  \Big\{ k-1 \\ & +\max\{t_j \, : \, h_j=i, \, \eta_j \neq 0 \} \, : \, i \in \{h_j \, : \, \eta_j \neq 0\} \Big\}.
\end{align*}
Thus, the dimension of the Schur square satisfies %
\begin{align*}
\dim \! \left(\Cmult^2\right) \geq \left|\big\{ d_1 + d_2 \, : \, d_1,d_2 \in S_{\tVec,\hVec,\etaVec}^{n,k}, \, d_1 + d_2<n \big\}\right|.
\end{align*}
\end{proposition}

\begin{proof}
Recall $g_i = X^i + \sum_{j=1, \, h_j=i}^{\ell} \eta_j X^{k-1+t_j}$ from \eqref{eq:evpolys_canonical_basis}. Hence, for $i \notin \{h_j \, : \, \eta_j \neq 0\}$, we have $g_i = X^i$ and otherwise, its degree is determined by the term $\eta_j X^{k-1+t_j}$ with largest $t_j$ among those $j$ with $h_j=i$ and $\eta_j \neq 0$.
The second part follows directly from Lemma~\ref{lem:Schur_square_evaluation_code}.
\end{proof}

Lemma~\ref{lem:Schur_square_evaluation_code} and Proposition~\ref{prop:trs_Schur_squares_TRS_lower_bound} imply the following three inequivalence statements for \startw and \plustw codes.

\begin{corollary}
Let $3 \leq k < \tfrac{n}{2}$. Then, any \startw code is non-GRS.
If $k=\tfrac{n}{2}$ and $\eta_1^2 \prod_{i=1}^{n} \alpha_i \neq 1$, then any \startw code with such $\etaVec$ and $\alphaVec$ is non-GRS. %
\end{corollary}

\begin{proof}
For a \startw code, we have $S_{\tVec,\hVec,\etaVec}^{n,k} = \{1,\dots,k\}$, so the set $A := \left\langle fg \, : \, f,g \in \Pset \right\rangle$ contains polynomials of degrees $\{2,\dots,2k\}$.
Furthermore, $A$ contains a polynomial of degree $1$ since $X^1 \cdot (\eta_1 X^k+X^0) -\eta_1 X^{k-1} \cdot X^2 = X^1$ (here we need $k\geq 3$).
In fact, we can choose as a basis of $A$ the polynomials $X^1,X^2,\dots,X^{2k-1},(1+\eta_1 X^k)^2$ since $X^i = X^{i_1} X^{i_2}$ for some $1 \leq i_1,i_2 \leq k-1$ for any $i=2,\dots,2k-2$, and $X^{2k-1} = \eta_1^{-1} \cdot X^{k-1} \cdot (\eta_1 X^k+X^0) - \eta_1^{-1} \cdot X^1 \cdot X^{k-2}$ (here we need $k \geq 3$).

If $2k<n$, the set $A$ hence contains $2k$ polynomials of distinct degrees less than $n$, and by Lemma~\ref{lem:Schur_square_evaluation_code}, we have $\dim \Code^2 \geq 2k$. In particular, $\Code$ is non-GRS.

If $2k=n$, then we must reduce the basis polynomial $(1+\eta_1 X^k)^2$ modulo $\prod_{i=1}^{n}(X-\alpha_i)$ in order to determine the Schur square dimension. As the monomials $X^1,X^2,\dots,X^{n-1}$ are in $A$, the Schur square has dimension $n$ if and only if the constant term of
\begin{align*}
\overline{(1+\eta_1 X^k)^2} = (1+\eta_1 X^k)^2 - \eta_1^2\prod_{i=1}^{n}(X-\alpha_i)
\end{align*}
is non-zero.
\end{proof}

\begin{corollary}
Let $3 \leq k < \tfrac{n}{2}$
Then, any \plustw code $\Code$ is non-GRS.
\end{corollary}

\begin{proof}
We have $S_{\tVec,\hVec,\etaVec}^{n,k} = \{0,1,\dots,k-2,k\}$ and thus $\big\{ d_1 + d_2 \, : \, d_1,d_2 \in S_{\tVec,\hVec,\etaVec}^{n,k}, \, d_1 + d_2<n \big\} = \{0,\dots,2k-2,2k\}$ (here we use $k \geq 3$ and $2k<n$).
By Proposition~\ref{prop:trs_Schur_squares_TRS_lower_bound}, we have $\Code^2 \geq 2k$ and the claim follows.
\end{proof}

For \startw codes with evaluation points forming a multiplicative group, we can use the duality statements of Section~\ref{sec:duals} and show that also high-rate codes are non-GRS.

\begin{corollary}
Let $\tfrac{n}{2} < k \leq n-3$ and suppose the $\alphaVec$ form a proper subgroup of $(\F_q^*, \cdot)$. Then, any \startw code with evaluation points $\alphaVec$ is non-GRS.
\end{corollary}

\begin{proof}
By \cref{thm:trs_duals}, the dual code of the \startw code is equivalent to a low-rate twisted RS code $\Code[n,n-k]$ with $\ell=1$, $t_1=k$, $h_1=n-k-1$, and the same evaluation points. Hence, the evaluation polynomial set $\Pset$ of $\Code$ is spanned by the polynomials $X^0,\dots,X^{n-k-2},\eta' X^{n-1}+X^{n-k-1}$ for some $\eta' \neq 0$.
As $\Code$ is a low-rate code, it suffices to show that $\dim \Code^2 \geq 2(n-k)$.
We show this by finding $2(n-k)$ polynomials of distinct degrees in $B := \left\langle \overline{fg} \, : \, f,g \in \Pset \right\rangle$ and applying Lemma~\ref{lem:Schur_square_evaluation_code}.

First note that by combining the basis elements, $B$ obviously contains elements of degrees $0,\dots,2(n-k)-4$ and $n-1$.
We construct two more elements of $B$ with differnet degrees using the structure of $\alphaVec$.
Since $\prod_{i=1}^{n} (X-\alpha_i) = X^n-1$,
the set $B$ contains a polynomial of degree $2(n-k)-3$ as
\begin{align*}
\overline{X^{n-k-2} \left(\eta' X^{n-1}+X^{n-k-1}\right)} = \eta' X^{n-k-3} + X^{2(n-k)-3}
\end{align*}
(we use that $n-k-2\geq 1$ due to $k \leq n-3$) and a polynomial of degree $n-2$ as
\begin{align*}
&\overline{\left(\eta' X^{n-1}+X^{n-k-1}\right)^2} \\
&\qquad \qquad = {\eta'}^2 X^{n-2} + 2 \eta'X^{n-k-2} + X^{2(n-k)-2}.
\end{align*}
Note that in the latter polynomial, $X^{n-2}$ is indeed the leading term due to $\tfrac{n}{2} < k$. For the same reason, we have $2(n-k)-3<n-2$. This concludes the proof.
\end{proof}

The following theorem shows that many twisted RS codes of rate smaller than $1/2$ are not GRS codes.
The only restriction is a mild technical condition on the hook vector $\hVec$, which we require to not contain the two smallest or the two largest possible values, or contain consecutive elements. %

\begin{theorem}
\label{thm:trs_non-GRS_low-rate_single_twisted_RS_codes}
Let $k <\tfrac{n}{2}$ and choose $\alphaVec,\hVec,\tVec,\etaVec$ as in \cref{def:twisted_RS_codes} with the additional requirements $\eta_i \neq 0$, $1 < h_i < k-2$, and either $h_i=h_j$ or $|h_i-h_j|>1$ for all $i\neq j$.
Then, the code $\Code := \Cmult$ has Schur square dimension $\dim\!\left(\Code^2\right) \geq 2k$.
In particular, it is not a GRS code.
\end{theorem}

\begin{proof}
By \cref{prop:trs_Schur_squares_TRS_lower_bound}, the set of evaluation polynomial degrees is given by $S_{\tVec,\hVec,\etaVec}^{n,k} = A \cup B$, where
\begin{align*}
A &= \{0,\dots,k-1\} \setminus \{h_1,\dots,h_\ell\} \\
\emptyset \neq B &\subseteq \{ k-1+t_i \, : \, i=1,\dots,\ell\}.
\end{align*}
By the restrictions on $h_i$, we have
\begin{align*}
\{0,1,k-2,k-1\} &\subseteq A &&\text{and} \\
\{h_i-1,h_i+1\} &\subseteq A &&\forall \, i=1,\dots,\ell.
\end{align*}
We show that $\{0,\dots,2k-2,\mu\} \subseteq \DsetText{\evpolys}$ for some $\mu \in \{2k-1,\dots,n-1\}$.
Let $0 \leq j \leq k-1$. Then, $j$ can be written as the sum of two elements in $A$ as follows:
\begin{equation*}
j = \begin{cases}
j+0, &\text{if } j \in A \text{ (i.e., $j \neq h_i$ for all $i$)}, \\
(h_i-1)+1, &\text{if } j = h_i.
\end{cases}
\end{equation*}
Hence, $j \in \DsetText{\evpolys}$. We used $0,1,h_i-1 \in A$.
Let $k \leq k-1+j \leq 2k-2$. Then,
\begin{equation*}
k-1+j = \begin{cases}
(k-1)+j, &\text{if } j \in A,\\
(k-2)+(h_1+1), &\text{if } j = h_i,
\end{cases}
\end{equation*}
i.e., $k-1+j \in \DsetText{\evpolys}$. We used $k-1,k-2,h_i+1 \in A$.

It is left to show that $\DsetText{\evpolys} \cap \{2k-1,\dots,n-1\}$ is non-empty. We distinguish three cases, of which at least one is true since $B$ is non-empty and $k \leq b < n$ for all $b \in B$:
\begin{enumerate}
\item If there is a $b\in B$ with $b \geq 2k$, then $0+b \in \DsetText{\evpolys}$ (recall that $0 \in A$) and the claim follows.
\item If $k \in B$, then $2k-1 = k + (k-1) \in \DsetText{\evpolys}$ (recall that $k-1 \in A$) and the claim follows.
\item If there is a $b \in B$ with $k < b < 2k$, then
\begin{align*}
&\DsetText{\evpolys} \cap \{2k-1,\dots,n-1\} \\
&\supseteq (b+A) \cap \{2k-1,\dots,n-1\} \\
&= \underbrace{\big\{ \max\{2k-1,b\},\dots,\min\{n-1,b+k-1\} \big\}}_{=: \, B_1} \setminus \\
&\quad \; \underbrace{\{b + h_i \, : \, i=1,\dots,\ell\}}_{=: \, B_2}
\end{align*}
Due to $\max\{2k-1,b\}=2k-1$, $2k<n$, and $b+k-1>2k-1$, we have $\{2k-1,2k\} \subseteq B_1$.
Since the $h_i$ are non-consecutive, we must have $2k-1 \notin B_2$ or $2k \notin B_2$. Hence, $B_1 \setminus B_2 \neq \emptyset$, which proves the claim.
\end{enumerate}
Hence, $|\DsetsmallernText{\evpolys}| \geq 2k$ and \cref{lem:Schur_square_evaluation_code} implies the claim.
\end{proof}

\subsection{A Combinatorial Inequivalence Argument}\label{ssec:inequivalence_combinartorial}

In the following, we present combinatorial results on the inequivalence question.
We rely on the following well-known characterization of GRS codes.

\begin{lemma}[\cite{roth1985generator,roth1989mds}]\label{lem:RS_characterization_minors}
Let $\Code$ be a linear code with a generator matrix of the form $\Gmat = [\I \mid \Amat]$.
Then, $\Code$ is a GRS code if and only if, for $\Atilde \in \Fq^{k \times n-k}$ with $A'_{ij} = A_{ij}^{-1}$,
\begin{enumerate}[label=(\roman*)]
\itemsep0pt
\item\label{itm:GRS_characterization_1} all entries of $\Amat$ are non-zero,
\item\label{itm:GRS_characterization_2} all $2 \times 2$ minors of $\Atilde$ are non-zero, and
\item\label{itm:GRS_characterization_3} all $3 \times 3$ minors of $\Atilde$ are zero.
\end{enumerate}
\end{lemma}

An MDS code always has a systematic generator matrix $\Gmat = [\I \mid \A]$ and fulfills Conditions~\ref{itm:GRS_characterization_1} and \ref{itm:GRS_characterization_2}.
The crucial difference of a GRS and a non-GRS MDS code is hence Condition~\ref{itm:GRS_characterization_3}.
Note also that for $\min\{k,n-k\} < 3$, the matrix $\Amat'$ has no $3 \times 3$ minors, so any such MDS code is a GRS code.
The following lemma states how the entries of $\Amat$ depend on $\etaVec$.

\begin{lemma}\label{lem:combinatorial_inequivalence}
Let $\alphaVec,\tVec,\hVec$ be chosen as in Definition~\ref{def:twisted_RS_codes}.
For these choices, let $\etaSet \subseteq \Fq^\numTwists$ be a set of $\etaVec$'s such that $\Cmult$ is MDS.
For any $\etaVec \in \etaSet$, let $\Gmat^{(\mathrm{sys}, \etaVec)} = [\I \mid \AetaM]$ be the systematic generator matrix of $\Cmult$. Then, the entries of $\AetaM \in \Fq^{k \times n-k}$ can be written as
\begin{equation*}
\AetaV_{i,j} = \frac{\pij{i,j}(\eta_1,\dots,\eta_\numTwists)}{p(\eta_1,\dots,\eta_\numTwists)} \quad \forall \, \etaVec = [\eta_1,\dots,\eta_\numTwists] \in \etaSet,
\end{equation*}
where $p, \pij{i,j} \in \Fq[X_1,\dots,X_\numTwists]$ are polynomials in $\numTwists$ variables of degree at most $1$ in each variable that do not have a zero in $\etaSet$ and whose coefficients do not depend on $\etaVec$.
\end{lemma}

\begin{proof}
Consider the ``canonical'' generator matrix $\Gmat^{(\mathrm{can}, \etaVec)}$ in~\eqref{eq:generator_matrix_canonical}, i.e., the matrix in which the rows are the evaluations at $\alphaVec$ of the evaluation polynomial basis $g_0,\dots,g_{k-1}$ (cf.~Lemma~\ref{lem:evpolys_basis}).
By definition of the $g_i$, its entries are of the form
\begin{equation}
G^{(\mathrm{can}, \etaVec)}_{i,j} = \alpha_j^{i-1} + \sum_{\substack{\kappa=1 \\ h_\kappa=i-1}}^{\ell} \eta_\kappa \alpha_j^{k-1+t_\kappa}, \label{eq:structure_systematic_G_polynomials}
\end{equation}
i.e., $G^{(\mathrm{can}, \etaVec)}_{i,j}$ is the evaluation at $\etaVec$ of a polynomial in $\Fq[X_1,\dots,X_\numTwists]$ of total degree at most $1$.
Furthermore, for each variable $X_i$, there is only one row of $\Gmat^{(\mathrm{can}, \etaVec)}$ for which these polynomials have non-zero degree in $X_i$ (we abbreviate the latter property with $X_i$ ``appears in a polynomial'' below).

We write $\Gmat^{(\mathrm{can}, \etaVec)} = [\Bmat^{(\etaVec)} \mid \Dmat^{(\etaVec)}]$ with $\Bmat^{(\etaVec)} \in \Fq^{k \times k}$ and $\Dmat^{(\etaVec)} \in \Fq^{k \times (n-k)}$.
Observe that since we only consider $\etaVec$ for which the code is MDS, $\Bmat^{(\etaVec)}$ is invertible and we have
\begin{equation*}
\Amat^{(\etaVec)} = {\Bmat^{(\etaVec)}}^{-1} \Dmat^{(\etaVec)} = \frac{\mathrm{adj}(\Bmat^{(\etaVec)}) \Dmat^{(\etaVec)}}{\det(\Bmat^{(\etaVec)})},
\end{equation*}
where $\mathrm{adj}(\Bmat^{(\etaVec)})$ is the adjunct matrix of $\Bmat^{(\etaVec)}$.

The determinant $\det\big(\Bmat^{(\etaVec)}\big)$ is the evaluation at $\etaVec$ of a polynomial $p \in \Fq[X_1,\dots,X_\numTwists]$, where $p$ has degree at most $1$ in each variable.
This follows inductively from the Laplace expansion of the determinant and the fact that each $X_i$ appears only in the polynomials that correspond to one row of $\Gmat^{(\mathrm{can}, \etaVec)}$.
Further, $p$ has no zero in $\etaSet$ since the code is MDS for all $\etaVec \in \etaSet$.
This gives the sought polynomial $p$.

We study the entries of $\Amat^{(\etaVec)} \cdot \det(\Bmat^{(\etaVec)})$, which are
sums of products of one entry from the adjunct matrix and one entry from $\Dmat^{(\etaVec)}$ whose column and row index, respectively, coincide.
By definition, the $(i,j)$-th entry of the adjunct matrix of $\Bmat^{(\etaVec)}$ is given by $(-1)^{i+j}$ times the determinant of the $(k-1) \times (k-1)$ submatrix of $\Bmat^{(\etaVec)}$ obtained by deleting its $i$-th column and $j$-th row.
This means that it is the evaluation at $\etaVec$ of a polynomial in $\Fq[X_1,\dots,X_\ell]$ with degree at most $1$ in each variable.
Moreover, if $\kappa$ satisfies $h_\kappa=j-1$, the variable $X_\kappa$ does not appear in those polynomials that correspond to the $j$-th column of the adjunct matrix.
On the other hand, these $X_\kappa$ are the only variables that appear in the polynomials corresponding to the $j$-th row of $\Dmat^{(\etaVec)}$ (cf.~\eqref{eq:structure_systematic_G_polynomials}).
Hence, the $(i,j)$-th entry of $\Amat^{(\etaVec)} \cdot \det(\Bmat^{(\etaVec)})$ can be written as evaluation at $\etaVec$ of a polynomial $p^{(i,j)} \in \Fq[X_1,\dots,X_\ell]$ of degree at most $1$ in each $X_i$. Furthermore, each of the $p^{(i,j)}$s does not have a zero in $\etaSet$ since otherwise $\Gmat^{(\mathrm{sys}, \etaVec)}$ would contain a row with $k$ zeros, contradicting the MDS assumption.
This gives the sought polynomials $p^{(i,j)}$.
\end{proof}

\begin{theorem}\label{thm:combinatorial_inequivalence}
Let $\min\{k,n-k\} \geq 3$ and $\alphaVec$, $\tVec$, $\hVec$ be chosen as in \cref{def:twisted_RS_codes}.
Denote by $\etaSet \subseteq \Fq^\numTwists$ the set of $\etaVec$ such that $\Cmult$ is MDS and assume that there is an $\etaVec^* \in \etaSet$ for which $\TRS{\alphaVec,\tVec,\hVec,\etaVec^\ast}{n,k}{}$ is non-GRS.
Then there is a non-zero multivariate polynomial $P \in \Fq[X_1,\dots,X_\numTwists]$ of degree at most $6$ in each variable such that all $\etaVec \in \etaSet$ for which $\Cmult$ is GRS are zeros of $P$.\footnote{In the first conference paper about twisted RS codes \cite{beelen2017twisted} (case $\ell=1$), we mistakenly assumed that the polynomial $P$ never vanishes. Hence, \cite[Theorem~18]{beelen2017twisted} is not true in general, see Example~\ref{ex:always_GRS_twisted_RS_codes}.}
\end{theorem}

\begin{proof}
Consider the systematic generator matrices $\Gmat^{(\mathrm{sys}, \etaVec)} = [\I \mid \AetaM]$ for all the codes indexed by $\etaVec$.
By Lemma~\ref{lem:RS_characterization_minors}, the code $\Cmult$ is GRS if and only if all $3 \times 3$ minors of the element-wise inverse of $\AetaM$ vanish.
Since there is an $\etaVec^* \in \etaSet$ such that $\TRS{\alphaVec,\tVec,\hVec,\etaVec^\ast}{n,k}{}$ is non-GRS, there is at least one non-zero minor of the element-wise inverse of $\Amat^{(\etaVec^*)}$.
Fix this minor (i.e., the same $3 \times 3$ submatrix) for all $\etaVec$.
We show that the $\etaVec$ for which this minor is zero are zeros of a polynomial $P$ as in the theorem statement.

By \cref{lem:combinatorial_inequivalence}, the entries of the element-wise inverse of $\Amat^{(\etaVec)}$ are evaluations at $\etaVec$ of rational functions $\tfrac{p}{p^{(i,j)}} \in \Fq(X_1,\dots,X_\ell)$, where $p,p^{(i,j)}$ are $\ell$-variate polynomials of degree at most $1$ in each variable which do not have a zero in $\etaSet$.
Hence, the fixed minor of the element-wise inverse of $\Amat^{(\etaVec)}$ is the evaluation at $\etaVec$ of a rational function $p^3 \tfrac{P}{Q}$, where $Q$ is the product of all nine $p^{(i,j)}$ in the $3 \times 3$ submatrix and $P$ is a sum of products of six $p^{(i,j)}$'s each. Thus, $P$ is a polynomial of degree at most $6$ in each variable.
As $Q$ and $p$ do not have a zero in $\etaSet$, the minor can only vanish at zeros of $P$.
Since $P(\etaVec^*) \neq 0$, the polynomial $P$ is non-zero.
\end{proof}

Theorem~\ref{thm:combinatorial_inequivalence} states that for given $\alphaVec$, $\tVec$, and $\hVec$, either all MDS twisted RS codes are GRS, or ``many'' are non-GRS.
We will quantify what we mean by ``many'' in the following, but first we give an example for which the polynomial $P$ in the proof of Theorem~\ref{thm:combinatorial_inequivalence} vanishes, i.e., all MDS twisted RS codes of this $\alphaVec$, $\tVec$, and $\hVec$ are GRS.

\begin{example}\label{ex:always_GRS_twisted_RS_codes}
Consider a twisted RS code over a field $\Fq$ with
\begin{equation*}
[n,k] = [6,3], \quad \ell = 1, \quad \tVec = 1, \quad \hVec = 2,
\end{equation*}
and evaluation points $\alphaVec = [\alpha_1,\dots,\alpha_6]$.
Let $\etaSet$ be the set of all $\etaVec$ such that the code $\Cmult$ is MDS.
Then, using the notation as in the proof of Lemma~\ref{lem:combinatorial_inequivalence}, the determinant of the $3 \times 3$ matrix $\Bmat^{(\etaVec)}$ is the evaluation at $\etaVec$ of the polynomial
\begin{align*}
p(X) = -\left[1+\left(\alpha_1+\alpha_2+\alpha_3\right)X\right]\prod_{\substack{i,j=1\\ i<j}}^{3}(\alpha_i-\alpha_j).
\end{align*}
The polynomials $p^{(i,j)}(X)$ ($i,j \in \{1,2,3\}$) as in Lemma~\ref{lem:combinatorial_inequivalence} are given as
\begin{align*}
p^{(i,j)}(X) &= \left(\alpha_{i_*}-\alpha_{j+3}\right)\left(\alpha_{i^*}-\alpha_{j+3}\right)\left(\alpha_{i_*}-\alpha_{i^*}\right) \\
& \quad \;  \cdot \left[1+\left(\alpha_{i_*}+\alpha_{i^*}+\alpha_{j+3}\right)X\right]
\end{align*}
where, or $i \in \{1,2,3\}$, we set $i_* := \min(\{1,2,3\} \setminus \{i\})$ and $i^* := \max(\{1,2,3\} \setminus \{i\})$.
As shown in Lemma~\ref{lem:combinatorial_inequivalence}, these polynomials are all non-zero (since the $\alpha_i$ are distinct) and of degree at most $1$.

Using the notation of the proof of Theorem~\ref{thm:combinatorial_inequivalence}, the determinant of the entry-wise inverse of $\Amat^{(\etaVec)}$ (note that $\Amat^{(\etaVec)}$ has only one $3 \times 3$ minor: the entire matrix) is the evalation at $\etaVec$ of the rational function $p^3 \tfrac{P}{Q}$, where
\begin{align*}
P &= - X^3 \left[1+X \sum_{i=1}^{3} \alpha_i \right]^2 \left[2+X \sum_{i=1}^{6} \alpha_i \right] \prod_{\substack{i,j=4\\ i<j}}^{6}(\alpha_i-\alpha_j), \\
Q &= \prod_{\substack{i,j,\kappa=1 \\ i<j}}^{3}\left[1+X\left(\alpha_i+\alpha_j+\alpha_{\kappa+3}\right)\right]\prod_{i,j=1}^{3} \left(\alpha_i-\alpha_{j+3}\right).
\end{align*}
Observe that $\deg P \leq 6$ and $\deg Q = \tbinom{3}{2}3 = 9$. Furthermore, the polynomial $Q$ has no zero in $\etaSet$ since each factor $\left[1+X\left(\alpha_i+\alpha_j+\alpha_{\kappa+3}\right)\right]$ is also a factor of the polynomial whose evaluation at $\etaVec$ is the determinant of the $3 \times 3$ submatrix of $\Gmat^{(\mathrm{can}, \etaVec)}$ consisting of the columns indexed by $i$, $j$, and $\kappa+3$ (which must be non-zero due to the MDS property).

By the same argument, the factor $1+X \sum_{i=1}^{3} \alpha_i$ of $P$ cannot have a zero in $\etaSet$.
The factor $X^3$ has only $\eta=0$ as a zero, which obviously yields a GRS code.
Hence, the code $\Cmult$ with $\etaVec \in \etaSet \setminus \{0\}$ is non-GRS if and only if $\etaVec$ is a zero of $2+X \sum_{i=1}^{6} \alpha_i$.
In particular, $P$ is the zero polynomial if and only if 
\begin{enumerate}
\item \label{itm:MDS_nonGRS_example_1} $\Fq$ has characteristic $2$ and
\item \label{itm:MDS_nonGRS_example_2} $\sum_{i=1}^{6} \alpha_i = 0$.
\end{enumerate}
This implies a few interesting observations:
\begin{itemize}
\item Since the second condition can be satisfied for $q = {2^m}$ if and only if $m \geq 4$, this gives a family of twisted RS codes with non-trivial parameters that are GRS for all $\etaVec \in \etaSet$.
\item If $\sum_{i=1}^{6} \alpha_i = 0$, but the characteristic of $\Fq$ is not $2$, then \emph{any} $\Cmult$ with $\etaVec \in \etaSet \setminus \{0\}$ is non-GRS. %
\item If the characteristic is $2$, but $\sum_{i=1}^{6} \alpha_i \neq 0$, then \emph{any} $\Cmult$ with $\etaVec \in \etaSet \setminus \{0\}$ is non-GRS. %
\item If the characteristic is not 2 and $\sum_{i=1}^{6} \alpha_i \neq 0$, then there is at most one $\etaVec \in \etaSet \setminus \{0\}$ such that the code $\Cmult$ is GRS.
\end{itemize}

\end{example}

\cref{thm:combinatorial_inequivalence} can be interpreted as follows: for fixed $n$, $k$, $\tVec$, $\hVec$, and $\alphaVec$, either \emph{all} $\etaVec$ corresponding to MDS codes are GRS, or only a number of them that is bounded by the number of roots of a non-zero $\ell$-variable polynomial of degree at most $6$ in each variable.

\begin{lemma}\label{lem:non_GRS_polynomial_roots}%
Let $P \in \Fq[X_1,\dots,X_\numTwists] \setminus \{0\}$ be a non-zero multivariate polynomial of degree at most $6$ in each variable, and $\etaSet = \etaSet_1 \times \dots \times \etaSet_\numTwists$, where $\etaSet_i \subseteq \Fq$ with $|\etaSet_i|>6$ for all $i$.
Then, $P$ has at most $\prod_{i=1}^{\numTwists} |\etaSet_i| - \prod_{i=1}^{\numTwists} (|\etaSet_i|-6)$ zeros in $\etaSet$.
\end{lemma}

\begin{proof}
The evaluation of $P$ at all elements of $\etaSet$ gives a codeword of an $\ell$-fold product code of GRS codes of parameters $[n_i,k,d_i]$, where $n_i := |\etaSet_i|$, $k=7$, and $d_i = |\etaSet_i|-6$. It is well-known that such a code has length $n = \prod_{i=1}^{\ell} n_i$ and minimum distance $d = \prod_{i=1}^{\ell} d_i$, so any non-zero codeword has weight at least $d$. Hence, $P$ has at most $n-d$ zeros in $\etaSet$, which gives the claim.
\end{proof}

\cref{thm:combinatorial_inequivalence} and \cref{lem:non_GRS_polynomial_roots} imply the following corollary.

\begin{corollary}
Let $n$, $k$, $\tVec$, $\hVec$, and $\alphaVec$ be chosen as in \cref{def:twisted_RS_codes} such that there are sets $\etaSet_i \subseteq \Fq$ with $|\etaSet_i| > 6$ and $\Cmult$ is MDS for any $\etaVec \in \etaSet := \etaSet_1 \times \dots \times \etaSet_\numTwists$.
Then, either
\begin{itemize}
\item all $\Cmult$ with $\etaVec \in \etaSet$ are GRS codes or 
\item $\Cmult$ is a non-GRS MDS twisted RS code for at least a fraction
$
A := \prod_{i=1}^{\ell} \big(1-\frac{6}{|\etaSet_i|}\big)
$
of the elements $\etaVec$ in $\etaSet$.
\end{itemize}
In particular, for the MDS constructions in Section~\ref{ssec:MDS_general}:
\begin{itemize}
\item \cref{prop:MDS_condition_subfield_chain}: we have $\etaSet_i := \F_{q_i} \setminus \F_{q_{i-1}}$, hence, for $q_0 \geq 4$, %
\begin{equation*}
A = \textstyle\prod_{i=1}^{\ell}\big(1-\tfrac{6}{q_i-q_{i-1}}\big) \geq \big(1-\tfrac{6}{n(n-1)}\big)^\ell.
\end{equation*}
\item \cref{prop:MDS_condition_power_basis}: we have $\etaSet_i := \{a \psi \, : \, a \in \Fqsmall^*\}$, hence, for $q_0 \geq 8$, %
\begin{equation*}
A = \big(1-\tfrac{6}{q_0-1}\big)^\ell \geq \big(1-\tfrac{6}{n-1}\big)^\ell.
\end{equation*}
\end{itemize}
\end{corollary}

The first conference version of this paper \cite{beelen2017twisted} contains several computer search results for twisted RS codes.
Among others, we counted inequivalent MDS twisted RS codes and non-GRS twisted RS codes for small parameters ($q \leq 13$).
The results show that for these parameters, most MDS twisted RS codes are non-GRS and there are also several parameters resulting in mutually inequivalent twisted RS codes.
We also compared twisted RS codes to Roth-Lempel codes \cite{roth1989mds}, whose definition is similar to our \plustw codes. The computer searches for small parameters show that the two code families are largely independent, i.e., only few of their equivalence classes intersect.
More details and tables can be found in \cite{beelen2017twisted}.

\section{Decoding}\label{sec:decoding}

Twisted RS codes can be decoded using a simple but expensive strategy: Use brute force to determine the twist coefficients $f_{h_i}$ for all $i=1,\dots,\numTwists$, for each choice subtract the evaluation of $\sum_{j=1}^{\numTwists} \eta_j f_{h_j} X^{k-1+t_j}$ from the received word, and decode in the corresponding Reed--Solomon code. This way, we obtain a decoder with complexity $q^\numTwists$ times the complexity of the RS decoder, and decoding radius equal to the used RS decoder. Note that the output list size is bounded by generic bounds on the list size, i.e., not necessarily exponential in $\numTwists$.

In this section, we present a decoding strategy that is often faster than this brute-force decoder.
This comes at the cost that we cannot rigorously prove that decoding works for any error vector up to the maximal decoding radius.
However, we present a variety of numerical results that indicate that the new decoder, for large decoding parameter, is able to decode up to almost half the minimum distance with overwhelming probability.

\subsection{Key Equations}

We fix a \emph{decoding parameter} $\zeta \in \ZZ_{\geq 0}$ and set up a system of key equations.
For notational convenience, we define $\Iset_\zeta := \big\{ \i \in \ZZ_{\geq 0}^\ell \, : \, \sum_{\mu=1}^{\ell} i_\mu \leq \decParam \big\}$ and, for $\mu \in \{1,\dots,\ell\}$, $\deltamu := [0,\dots,0,1,0,\dots,0]$ ($\mu$-th unit vector).
Note that
\begin{align}
|\Iset_{\zeta}| = \binom{\ell+\zeta}{\ell} \quad \text{and} \quad |\Iset_{\zeta+1}| = \tfrac{\ell+\zeta+1}{\zeta+1}|\Iset_{\zeta}|. \label{eq:size_Iset}
\end{align}

We assume that we are given a received word $\r = \c+\e \in \Fq^{n}$, where $\c := \ev{f}{\alphaVec}$, for $f \in \evpolys$, is a codeword of a twisted RS code $\Cmult$, and $\e \in \Fq^{n}$ is an error of Hamming weight $\wtH(\e)=t$ and support $\Eset = \supp(\e) := \{i \, : \, e_i \neq 0\}$.
We define the polynomials $\Lambda := \prod_{i \in \Eset} (X-\alpha_i)$ (error locator polynomial), $g := \sum_{i=0}^{k-1} f_i X^i$, where the $f_i$'s are the coefficients of $f$, $G := \prod_{i=1}^{n} (X-\alpha_i)$, and $R$ to be the unique polynomial of degree $<n$ with $R(\alpha_i) = r_i$ for all $i=1,\dots,n$ (interpolation polynomial of the received word).
Define
\begin{align*}
\Lambda_\i &:=  \Lambda \prod_{\mu=1}^{\ell} f_{h_\mu}^{i_\mu} && \forall \, \i \in \Iset_{\decParam+1}, \\
\Psi_\j &:=  \Lambda_\j g && \forall \, \j \in \Iset_\decParam.
\end{align*}

The following system of key equations relates the notions defined above. %

\begin{theorem}[Key Equations]\label{thm:key_equations}
Consider the setting and notation above. %
Then, we have for all $\i \in \Iset_\decParam$
\begin{align*}
\Lambda_\i R \equiv \Psi_\i + \sum_{\mu=1}^{\ell} \Lambda_{(\i+\deltamu)} \eta_\mu X^{k-1+t_\mu} \pmod G.
\end{align*}
Furthermore, we have
\begin{align*}
\deg \Lambda_\i &\leq \deg \Lambda_\0 &&\forall \, \i \in \Iset_{\decParam+1}, \\
\deg \Psi_\j &\leq \deg \Lambda_\0+k-1 &&\forall \, \j \in \Iset_{\decParam}.
\end{align*}
\end{theorem}

\begin{proof}
We have $\Lambda R \equiv \Lambda f \pmod G$ since $[\Lambda(R-f)](\alpha_i)=0$ for all $i=1,\dots,n$. By the structure of $f$, we thus have
\begin{align}
\Lambda R \equiv \Lambda g + \sum_{\mu=1}^{\ell} \Lambda f_{h_\mu} \eta_\mu X^{k-1+t_\mu} \pmod G. \label{eq:core_key_equation}
\end{align}
Multiplying \eqref{eq:core_key_equation} with $\prod_{\mu=1}^{\ell} f_{h_\mu}^{i_\mu}$ gives the result. The degree bounds follow immediately from the definition, $\deg g <k$, and the fact that the $f_{h_\mu}$ are scalars.
\end{proof}

Solving the system of key equations in \cref{thm:key_equations} for the unknowns $\Lambda$, $g$, and $f_{h_\mu}$ is a non-linear problem.
To find a solution efficiently, we linearize the problem as follows.

\begin{problem}\label{prob:pade_approx}
Given $\tVec$ and $\etaVec$, and let $\r$ be a received word. Denote by $R$ and $G$ the polynomials defined above \cref{thm:key_equations}. %
Find polynomials $(\lambda_\i)_{\i \in \Iset_{\decParam+1}}$ and $(\psi_\j)_{\j \in \Iset_{\decParam}}$, not all zero, such that
\begin{align}
\lambda_\i R &\equiv \psi_\i + \sum_{\mu=1}^{\ell} \lambda_{(\i+\deltamu)} \eta_\mu X^{k-1+t_\mu} \pmod G, \label{eq:pade_congruence}\\
\deg \lambda_\i &\leq \deg \lambda_\0, \label{eq:pade_deg_lambda} \\
\deg \psi_\i &\leq \deg \lambda_\0+k-1, \label{eq:pade_deg_psi}
\end{align}
 for all $\i \in \Iset_\zeta$. %
\end{problem}

For a solution $(\lambda_\i)_{\i \in \Iset_{\decParam+1}}$ and $(\psi_\j)_{\j \in \Iset_{\decParam}}$ of \cref{prob:pade_approx}, we call $\deg \lambda_\0$ the \emph{degree} of the solution.

The problem is related to the decoding problem as follows:
$(\lambda_\i = \Lambda_\i)_{\i \in \Iset_{\decParam+1}}$, $(\psi_\j = \Psi_\j)_{\j \in \Iset_{\decParam}}$ is a solution of \cref{prob:pade_approx} of degree $t$, where $t = \wtH(\e)$ is the number of errors.
As we want to find the error locator polynomial $\Lambda_\0$ of minimal degree (i.e., the one corresponding to an error of smallest weight), we aim at finding a solution of \cref{prob:pade_approx} of smallest-possible degree.
If all goes well and there are no generic solutions of the problem of equal
or smaller degree, then we indeed find the solution $(\lambda_\i = \Lambda_\i)_{\i \in \Iset_{\decParam+1}}$, $(\psi_\j = \Psi_\j)_{\j \in \Iset_{\decParam}}$ or a scalar multiple thereof.
We can thus obtain $g$, i.e., the lowest $k$ coefficients of the message polynomial, by division
\begin{equation*}
g = \frac{\psi_\0}{\lambda_\0}.
\end{equation*}

\begin{algorithm}
	\caption{Decoding Algorithm for Twisted RS Codes}
	\label{alg:decoder}
	\SetKwInOut{Input}{Input}\SetKwInOut{Output}{Output}
	\Input{Received Word $\r$, code $\Cmult$, and decoder parameter $\decParam$}
	\Output{A closest codeword $\c \in \Cmult$ to $\r$, or ``decoding failure''.}
	Compute $R$ and $G$ as defined above \cref{thm:key_equations} \label{line:R_and_G} \\ %
	$(\lambda_\i)_{\i \in \Iset_{\decParam+1}}, \, (\psi_\j)_{\j \in \Iset_{\decParam}} \gets$ solution of minimal degree of \cref{prob:pade_approx} with input $R$ and $G$. \label{line:solve_linearized_problem} \\
	\If{$\lambda_\0$ divides $\psi_\0$ \label{line:divisibility_check}}{
		$g \gets \psi_\0/\lambda_\0$ \label{line:division} \\
		$\c \gets \ev{\sum_{i=0}^{k-1} g_i X^i + \sum_{j=1}^{\numTwists} \eta_j g_{h_j} X^{k-1+t_j}}{\alphaVec}$ \label{line:evaluation} \\
		\If{$\dH(\c,\r) \leq \lfloor \tfrac{n-k}{2} \rfloor$} {
		\Return{$\c$}
		}
	}
	\Return{``decoding failure''}

\end{algorithm}
The resulting decoder is summarized in \cref{alg:decoder}.
Note that a minimal solution of \cref{prob:pade_approx} can be found by solving the linear system of equations for any $\tau = \deg \lambda_\0 = 0,1,2,\dots$ (w.l.o.g.\ we choose $\lambda_\0$ to be monic), implied by the congruence \eqref{eq:pade_congruence} and degree constraints \eqref{eq:pade_deg_lambda} and \eqref{eq:pade_deg_psi}, until a solution exists.
In \cref{app:efficiently_decoding}, we show that the decoding algorithm can be implemented more efficiently, more precisely with complexity
\begin{align}
O^\sim\!\left( \left(e\tfrac{\ell+\zeta+1}{\ell}\right)^{\ell \omega} n\right) \label{eq:complexity}
\end{align}
operations in $\Fq$, where $e$ is Euler's number.

\subsection{Decoding Radius}\label{ssec:decoding_radius}

The new decoder is a partial unique decoder,
which means that for some error weights, some error patterns cannot be corrected, but if it works, then it returns a unique decoding solution.
We informally call the maximal value $\tau$ up to which the decoder returns $\c$ from the input $\r = \c+\e$ for the majority of the error vectors $\e$ of weight $\tau$ the \emph{decoding radius} of the new decoder, and denote it by $\taumax$.

Although we have no failure probability bound or the like for the new decoding algorithm, we present some heuristic arguments in this section that lead to expected upper and lower bounds of the decoding radius.
Our numerical results in Section~\ref{ssec:decoding_numerical_results} verify the expectation on various examples, with only very few exceptions.

Recall that $(\lambda_\i = \Lambda_\i)_{\i \in \Iset_{\decParam+1}}$, $(\psi_\j = \Psi_\j)_{\j \in \Iset_{\decParam}}$ is a solution of \cref{prob:pade_approx} of degree $\tau$, where $\tau = \wtH(\e)$. Furthermore, $\lambda_\0$ is monic and of degree $\tau$.
Decoding succeeds if $(\lambda_\i = \Lambda_\i)_{\i \in \Iset_{\decParam+1}}$, $(\psi_\j = \Psi_\j)_{\j \in \Iset_{\decParam}}$ is the only solution of \cref{prob:pade_approx} of degree $\tau$ and monic $\lambda_\0$, and there is no solution of \cref{prob:pade_approx} of smaller degree. Note that the other direction is not necessarily true.

All solutions of \cref{prob:pade_approx} of degree exactly $\tau$ and monic $\lambda_\0$ can be determined by an inhomogeneous linear system of equations, where the unknowns are the coefficients of the $\lambda_\i$ (except for the leading term of $\lambda_\0$, which is set to $1$) and $\psi_\i$ (the number of coefficients, and thus unknowns is determined by the degree bounds \eqref{eq:pade_deg_lambda} and \eqref{eq:pade_deg_psi}) and whose equations are given by the congruence relations \eqref{eq:pade_congruence}.
This means that the system has $\mathrm{NE} = n |\Iset_{\decParam}|$ equations and $\mathrm{NV} = |\Iset_{\decParam+1}|(\tau+1)+(\tau+k)|\Iset_{\decParam}|-1$ variables. %
The matrix of the linear system is of the form
\begin{align}
\begin{bmatrix}
\Rmat & \I_{n \times (k+\tau)} & \A \\
\0 & \0 & \B
\end{bmatrix}, \label{eq:decoding_linear system}
\end{align}
where $\Rmat \in \Fq^{n \times \tau}$ (depends on the received word $\r$), $\A \in \Fq^{n \times (\mathrm{NV}-2\tau-k)}$, $\B = \Fq^{(\mathrm{NE}-n) \times (\mathrm{NV}-2\tau-k)}$ (depends on the received word $\r$), and $\I_{n \times (k+\tau)}$ is an $n \times (k+\tau)$ matrix with ones on the diagonal and zero otherwise.
The columns of the submatrix $\begin{bmatrix}
\Rmat \\ \0
\end{bmatrix}$ correspond to the coefficients $0,\dots,\tau-1$ of $\lambda_\0$ and the columns of $\begin{bmatrix}
\I_{n \times (k+\tau)} \\  \0
\end{bmatrix}$ correspond to the coefficients of $\psi_\0$.

The decoding radius corresponds to the maximal integer $\tau$ for which, for the majority of error vectors of weight $\tau$, the linear system of equations has a unique solution and no solution for smaller values of $\tau$.

The linear system has $\mathrm{NE} = n |\Iset_{\decParam}|$ equations and $\mathrm{NV} = |\Iset_{\decParam+1}|(\tau+1)+(\tau+k)|\Iset_{\decParam}|-1$ variables.
Hence, if $\tau$ is the number of errors, Problem~\ref{prob:pade_approx} has more than one solution for
\begin{align}
\mathrm{NE}+2 & \leq \mathrm{NV}  \notag \\
\Leftrightarrow \quad n |\Iset_{\decParam}|+3 &\leq |\Iset_{\decParam+1}|(\tau+1)+(\tau+k)|\Iset_{\decParam}| \notag\\
\Leftrightarrow \, \tau \big( |\Iset_{\decParam+1}| + |\Iset_{\decParam}| \big) &\geq (n-k) |\Iset_{\decParam}|-|\Iset_{\decParam+1}|+2  \notag\\
\Leftrightarrow \; \tau \geq \tfrac{|\Iset_{\decParam}|}{|\Iset_{\decParam+1}| + |\Iset_{\decParam}|}(n-k) &-\tfrac{|\Iset_{\decParam+1}|-3}{|\Iset_{\decParam+1}| + |\Iset_{\decParam}|}\notag \\
= \tfrac{\zeta+1}{2(\zeta+1)+\ell}(n-k) &-\tfrac{\zeta+\ell+1-3(\zeta+1)\binom{\ell+\zeta}{\ell}^{-1}}{2(\zeta+1)+\ell}, \label{eq:decoding_radius}
\end{align}
where we use \eqref{eq:size_Iset} to obtain the last line. 
Since the matrix $\Rmat$ and parts of the matrix $\B$ in \eqref{eq:decoding_linear system} depend on the received word $\r$ and appear to behave somewhat like random matrices for random errors, we expect that the decoder behaves as follows: for the majority of error vectors of weight $\tau$, for $\tau$ smaller than the right-hand side of \eqref{eq:decoding_radius}, the linear system has only one solution, $(\lambda_\i = \Lambda_\i)_{\i \in \Iset_{\decParam+1}}$, $(\psi_\j = \Psi_\j)_{\j \in \Iset_{\decParam}}$, and no solution of smaller degree.

Based on these observations, we expect that the decoding radius is at least as large as 
\begin{align}
\tauLB &:= \left\lceil\tfrac{\zeta+1}{2(\zeta+1)+\ell}(n-k) -\tfrac{\zeta+\ell+1-3(\zeta+1)\binom{\ell+\zeta}{\ell}^{-1}}{2(\zeta+1)+\ell} \right\rceil -1 \notag \\
&\approx \tfrac{\zeta+1}{2(\zeta+1)+\ell}(n-k). \label{eq:tauLB}
\end{align}

If we inspect the linear system closer, we observe that above the radius $\tauLB$, even if a solution $\lambda_\i$, $\psi_\j$ is not unique, the polynomials $\lambda_\0$ and $\psi_\0$ may be the same for all solutions.
Hence, in this case, Algorithm~\ref{alg:decoder} is able to retrieve the correct error positions from any solution.

Consider again the system matrix in \eqref{eq:decoding_linear system}.
If Problem~\ref{prob:pade_approx} has multiple solutions, but $\lambda_\0$ and $\psi_\0$ are the same for all of them, then
the rank of the entire matrix is less than $\mathrm{NV}$ (i.e., the number of columns), but
we have $\rank
\begin{bmatrix}
\Rmat & \I_{n \times (k+\tau)} \\ \0 & \0
\end{bmatrix} = k+2\tau$ and the column spaces of $\begin{bmatrix}
\Rmat \\ \0
\end{bmatrix}$ and $\begin{bmatrix}
	\Rmat & \I_{n \times (k+\tau)} \\ \0 & \0
\end{bmatrix}$ do not intersect.
It is quite involved to predict only from the code parameters $n,k,\tVec,\hVec,\etaVec,\alphaVec$ for which exact values of $\tau$ these properties are fulfilled with high probability, since the matrices $\Rmat$ and $\B$ depend on the received word $\r$.
However, it is clear that $\rank
\begin{bmatrix}
\Rmat & \I_{n \times (k+\tau)} \\ \0 & \0
\end{bmatrix} < k+2\tau$ for $\tau > \tfrac{n-k}{2}$, which gives an upper bound on the decoding radius.
In summary, we get the following conjecture.

\begin{expectation}\label{conj:decoding_radius}
	The decoding radius $\tau_\mathsf{max}$ of Algorithm~\ref{alg:decoder} satisfies
	\begin{align*}
	\tauLB \leq \tau_\mathsf{max} \leq \lfloor \tfrac{n-k}{2} \rfloor,
	\end{align*}
	where $\tauLB$ is defined as in \eqref{eq:tauLB}.
\end{expectation}
Our numerical results in Section~\ref{ssec:decoding_numerical_results} confirm this expectation for various parameters, with only very few exceptions.
Furthermore, the numerical results show that for randomly chosen errors of a given weight $\tau$, the success probability of the decoding is close to $1$ for $\tau$ up to the decoding radius, and close to $0$ above.
Note that, for a given $\varepsilon>0$, we may choose $\decParam \geq \tfrac{1-\varepsilon}{2\varepsilon} \ell -1$ and get
\begin{align*}
\tauLB \geq (1-\varepsilon) \tfrac{n-k}{2}.
\end{align*}
Hence, $\tauLB$
converges to $\lfloor \tfrac{n-k}{2} \rfloor$ for growing decoding parameter~$\decParam$.

Furthermore, for given $\varepsilon$, we can rewrite the decoding complexity expression of \eqref{eq:complexity} into
$O^\sim\!\left( \left(\tfrac{e}{2\varepsilon}\right)^{\ell \omega} n\right)$
operations in $\Fq$.
For comparison, a brute-force decoder for correcting the same number of errors costs $O^\sim\!\left( q^\ell n\right)$ operations in $\Fq$. Hence, the new decoder is faster for
\begin{align*}
\left(\tfrac{e}{2 \varepsilon}\right)^\omega \ll q,
\end{align*}
Note that the left-hand side does not depend on the code length $n$ or the field size $q$.

\subsection{Numerical Results}\label{ssec:decoding_numerical_results}

In the following, we present numerical results obtained through Monte-Carlo simulations, which verify the expectation on the decoding radius for a variety of code and decoder parameters.

\mysubsubsection{Monte-Carlo Simulations}

We consider the code parameters $q \in \{23,64,101\}$, $n=q-1$, code rates $\approx 0.3,0.5,0.7$, and number of twists $\ell =1,2,3$.
For a fixed parameter tuple $[q,n,k,\ell]$, we selected $50$ twisted RS codes at random in the following way:
\begin{itemize}
\item $\{\alpha_1,\dots,\alpha_n\}$ is chosen uniformly at random from the set of subsets of $\Fq^\ast$ of cardinality $n$.
\item $\tVec$, $\hVec$ is chosen uniformly at random from the set of valid twist/hook vectors with distinct entries, respectively.
\item $\etaVec$ is entry-wise chosen uniformly at random from~$\Fq^\ast$.
\end{itemize}
Note that these twisted RS codes are not necessarily MDS codes. In total, we created $1350$ random codes.

Then, for each such random code, we performed, for several decoding parameters $\decParam \in \{2,4,6\}$ and decoding radii $\tau \in \{\max\{0,\tauLB-2\}, \dots, \lfloor \tfrac{n-k}{2} \rfloor\}$, the following Monte-Carlo simulation:
\begin{itemize}
	\item Draw a codeword $\c$ of $\Cmult$ uniformly at random
	\item Draw an error $\e$ uniformly at random from the set of vectors of Hamming weight $\tau$
	\item Decode with decoding parameter $\decParam$
	\item If the decoder returns $\c$, declare a \emph{success}. Otherwise, declare a \emph{failure}\footnote{Note that this notion of failure includes the ``decoding failure'' declared (and noticed) by the decoder, as well as decoding errors (the decoder returns a valid codeword not equal to $\c$, also called miscorrections).}.
\end{itemize}
We performed this simulation $1000$ times for each parameters set and estimated the failure probability of the decoder for this code, $\decParam$ and radius $\tau$.
In total, we obtained $\approx 1.7 \cdot 10^7$ samples of the Monte-Carlo simulations.

\mysubsubsection{Tables}

Tables~\ref{tab:q=23}, \ref{tab:q=64}, and \ref{tab:q=101} contain the following information extracted from these Monte-Carlo simulations:
\begin{itemize}
\item For each code $\Cmult$ and each decoding parameter $\decParam$, we determine the decoding radius as the maximal value of $\tau$ for which the estimated failure probability is $<0.2$.
\item Each row of the table corresponds to a parameter set $[q,n,k,\ell]$ and decoding parameter $\decParam$.
We display the numbers of codes of the parameter set (out of $50$) which have a certain decoding radius $\taumax$.
\item In each row of the table, the entry below the expected lower bound on the decoding radius, $\taumax = \tau_{LB}$, is marked by a superscript $\mathsf{L}$, and similarly the entry below the upper bound $\taumax= \lceil \frac{n-k}{2}\rceil$ is marked by superscript $\mathsf{U}$. The cells corresponding to the expected range of the decoding radius have gray background color.
\item The table also contains the following three probabilities:
\begin{itemize}
	\item $\PfTwoBelow$ is the maximal observed failure probability one below the decoding radius of a code, i.e., at $\taumax-1$, maximized over all $50$ codes in this row.
	\item $\PfOneBelow$ is the maximal observed failure probability at the decoding radius of a code, i.e., at $\taumax$, maximized over all $50$ codes in this row.
	\item $\PfAbove$ is the minimal observed failure probability one above the decoding radius of a code, i.e., at $\taumax+1$, minimized over all $50$ codes in this row.
\end{itemize}
Note that we display only the ``worst'' probabilities (out of $50$ codes) for each row, and that the three probabilities may correspond to different codes.
\end{itemize}

\begin{table}[ht]
	\caption{Table for $q=23$ and $n=22$. See Section~\ref{ssec:decoding_numerical_results} for the description.}
	\label{tab:q=23}
	\setlength{\tabcolsep}{2pt}	
	\begin{center}
		\begin{tabular}{c|c|c||c|c|c|c|c|c|c|c||c|c|c}
			\multicolumn{3}{c||}{Para-} & \multicolumn{8}{c||}{Number of codes (out of $50$)} & \multicolumn{3}{c}{Observed} \\
			\multicolumn{3}{c||}{meters} & \multicolumn{8}{c||}{that have $\tau_\mathsf{max}=$} & \multicolumn{3}{c}{Failure Rates} \\
			$k$ & $\ell$ & $\zeta$ & $0$ & $1$ & $2$ & $3$ & $4$ & $5$ & $6$ & $7$ & $\PfTwoBelow$ & $\PfOneBelow$ & $\PfAbove$ \\
			\hline \hline
			$7$ & $1$ & $2$ &     &     &     &     & $0$ & $0$ & $43^\mathsf{L}$ \colorthiscell & $7^\mathsf{U}$ \colorthiscell & $0.000$ & $0.047$ & $0.905$ \\
			&     & $4$ &     &     &     &     & $0$ & $0$ & $43^\mathsf{L}$ \colorthiscell & $7^\mathsf{U}$ \colorthiscell & $0.000$ & $0.004$ & $0.915$ \\
			&     & $6$ &     &     &     &     & $0$ & $0$ & $43^\mathsf{L}$ \colorthiscell & $7^\mathsf{U}$ \colorthiscell & $0.000$ & $0.005$ & $0.905$ \\
			\hline
			& $2$ & $2$ &     &     &     & $0$ & $0$ & $42^\mathsf{L}$ \colorthiscell & $8$ \colorthiscell & $0^\mathsf{U}$ \colorthiscell & $0.000$ & $0.073$ & $0.890$ \\
			&     & $4$ &     &     &     & $0$ & $0$ & $29^\mathsf{L}$ \colorthiscell & $21$ \colorthiscell & $0^\mathsf{U}$ \colorthiscell & $0.001$ & $0.057$ & $0.878$ \\
			\hline
			& $3$ & $2$ &     &     & $0$ & $0$ & $36^\mathsf{L}$ \colorthiscell & $13$ \colorthiscell & $1$ \colorthiscell & $0^\mathsf{U}$ \colorthiscell & $0.000$ & $0.089$ & $0.861$ \\
			\hline
			$11$ & $1$ & $2$ &     &     & $0$ & $0$ & $40^\mathsf{L}$ \colorthiscell & $10^\mathsf{U}$ \colorthiscell &     &     & $0.000$ & $0.004$ & $0.909$ \\
			&     & $4$ &     &     & $0$ & $0$ & $40^\mathsf{L}$ \colorthiscell & $10^\mathsf{U}$ \colorthiscell &     &     & $0.000$ & $0.004$ & $0.906$ \\
			&     & $6$ &     &     & $0$ & $0$ & $40^\mathsf{L}$ \colorthiscell & $10^\mathsf{U}$ \colorthiscell &     &     & $0.000$ & $0.004$ & $0.906$ \\
			\hline
			& $2$ & $2$ &     & $0$ & $0$ & $31^\mathsf{L}$ \colorthiscell & $19$ \colorthiscell & $0^\mathsf{U}$ \colorthiscell &     &     & $0.000$ & $0.087$ & $0.896$ \\
			&     & $4$ &     &     & $0$ & $0$ & $50^\mathsf{L}$ \colorthiscell & $0^\mathsf{U}$ \colorthiscell &     &     & $0.007$ & $0.077$ & $0.955$ \\
			\hline
			& $3$ & $2$ &     & $0$ & $0$ & $47^\mathsf{L}$ \colorthiscell & $3$ \colorthiscell & $0^\mathsf{U}$ \colorthiscell &     &     & $0.005$ & $0.119$ & $0.921$ \\
			\hline
			$15$ & $1$ & $2$ & $0$ & $0$ & $36^\mathsf{L}$ \colorthiscell & $14^\mathsf{U}$ \colorthiscell &     &     &     &     & $0.000$ & $0.004$ & $0.900$ \\
			&     & $4$ & $0$ & $0$ & $36^\mathsf{L}$ \colorthiscell & $14^\mathsf{U}$ \colorthiscell &     &     &     &     & $0.000$ & $0.005$ & $0.907$ \\
			&     & $6$ & $0$ & $0$ & $36^\mathsf{L}$ \colorthiscell & $14^\mathsf{U}$ \colorthiscell &     &     &     &     & $0.000$ & $0.004$ & $0.900$ \\
			\hline
			& $2$ & $2$ & $0$ & $0$ & $50^\mathsf{L}$ \colorthiscell & $0^\mathsf{U}$ \colorthiscell &     &     &     &     & $0.000$ & $0.102$ & $0.958$ \\
			&     & $4$ & $0$ & $0$ & $50^\mathsf{L}$ \colorthiscell & $0^\mathsf{U}$ \colorthiscell &     &     &     &     & $0.000$ & $0.098$ & $0.957$ \\
			\hline
			& $3$ & $2$ & $2$ & $46^\mathsf{L}$ \colorthiscell & $2$ \colorthiscell & $0^\mathsf{U}$ \colorthiscell &     &     &     &     & $0.000$ & $0.000$ & $0.879$
		\end{tabular}
	\end{center}
\end{table}

\mysubsubsection{Observations}

It can be seen that for the vast majority of the codes, the decoding radius indeed lies between $\tauLB$ and $\lfloor \tfrac{n-k}{2} \rfloor$. This confirms our expectation that we derived heuristically in the previous subsection.
There are only very few exceptions: e.g., for
\begin{itemize}
	\item $[q,n,k,\ell] = [23,22,15,3]$ and $\decParam=2$
	\item $[q,n,k,\ell] = [64,63,19,2]$ and $\decParam=2$
	\item $[q,n,k,\ell] = [64,63,32,1]$ and $\decParam=2$
\end{itemize}
there are $2$, $1$, and $6$ codes, respectively, whose decoding radius is one below $\tauLB$. These are $9$ exceptions out of in total $2700$ code/decoding parameter pairs. Furthermore, in all exceptional cases, the decoding radius is only one below the expected smallest decoding radius.

From the values of $\PfTwoBelow$, $\PfOneBelow$, and $\PfAbove$, it can also be seen that the failure probability changes sharply around the decoding radius: for many parameters, the worst observed failure probability at the decoding radius is very small, e.g. $0.004$ for some codes. One below the decoding radius, the observed failure probability is $0$ for most parameters (recall that the number of samples is $1000$, so we can say that it is $\lessapprox 10^{-3}$ with some confidence). Above the decoding radius, the failure probability is always close to $1$, as expected.

\begin{table*}
\caption{Table for $q=64$ and $n=63$. See Section~\ref{ssec:decoding_numerical_results} for the description.}
\label{tab:q=64}
\begin{center}
\setlength{\tabcolsep}{2pt}	
\begin{tabular}{c|c|c||c|c|c|c|c|c|c|c|c|c|c|c|c|c|c|c|c|c|c|c||c|c|c}
	\multicolumn{3}{c||}{Parameters} & \multicolumn{20}{c||}{Number of codes (out of $50$) that have $\tau_\mathsf{max}=$} & \multicolumn{3}{c}{Observed Failure Rates} \\
	\,$k$\, & \,$\ell$\, & $\zeta$ & $3$ & $4$ & $5$ & $6$ & $7$ & $8$ & $9$ & $10$ & $11$ & $12$ & $13$ & $14$ & $15$ & $16$ & $17$ & $18$ & $19$ & $20$ & $21$ & $22$ & $\PfTwoBelow$ & $\PfOneBelow$ & $\PfAbove$ \\
	\hline \hline
	$19$ & $1$ & $2$ &     &     &     &     &     &     &     &     &     &     &     &     &     & $0$ & $0$ & $17^\mathsf{L}$ \colorthiscell & $11$ \colorthiscell & $15$ \colorthiscell & $7$ \colorthiscell & $0^\mathsf{U}$ \colorthiscell & $0.000$ & $0.034$ & $0.973$ \\
	    &     & $4$ &     &     &     &     &     &     &     &     &     &     &     &     &     &     & $0$ & $0$ & $17^\mathsf{L}$ \colorthiscell & $20$ \colorthiscell & $13$ \colorthiscell & $0^\mathsf{U}$ \colorthiscell & $0.001$ & $0.034$ & $0.976$ \\
	    &     & $6$ &     &     &     &     &     &     &     &     &     &     &     &     &     &     &     & $0$ & $0$ & $32^\mathsf{L}$ \colorthiscell & $18$ \colorthiscell & $0^\mathsf{U}$ \colorthiscell & $0.001$ & $0.037$ & $0.981$ \\
	\hline
	    & $2$ & $2$ &     &     &     &     &     &     &     &     &     &     &     & $0$ & $1$ & $12^\mathsf{L}$ \colorthiscell & $15$ \colorthiscell & $11$ \colorthiscell & $10$ \colorthiscell & $1$ \colorthiscell & $0$ \colorthiscell & $0^\mathsf{U}$ \colorthiscell & $0.000$ & $0.021$ & $0.971$ \\
	    &     & $4$ &     &     &     &     &     &     &     &     &     &     &     &     & $0$ & $0$ & $1^\mathsf{L}$ \colorthiscell & $21$ \colorthiscell & $16$ \colorthiscell & $12$ \colorthiscell & $0$ \colorthiscell & $0^\mathsf{U}$ \colorthiscell & $0.001$ & $0.031$ & $0.975$ \\
	\hline
	    & $3$ & $2$ &     &     &     &     &     &     &     &     &     & $0$ & $0$ & $12^\mathsf{L}$ \colorthiscell & $18$ \colorthiscell & $10$ \colorthiscell & $4$ \colorthiscell & $5$ \colorthiscell & $1$ \colorthiscell & $0$ \colorthiscell & $0$ \colorthiscell & $0^\mathsf{U}$ \colorthiscell & $0.000$ & $0.022$ & $0.950$ \\
	\hline
	$32$ & $1$ & $2$ &     &     &     &     &     &     &     &     & $0$ & $6$ & $17^\mathsf{L}$ \colorthiscell & $24$ \colorthiscell & $3^\mathsf{U}$ \colorthiscell &     &     &     &     &     &     &     & $0.000$ & $0.023$ & $0.979$ \\
	    &     & $4$ &     &     &     &     &     &     &     &     & $0$ & $0$ & $13^\mathsf{L}$ \colorthiscell & $34$ \colorthiscell & $3^\mathsf{U}$ \colorthiscell &     &     &     &     &     &     &     & $0.001$ & $0.029$ & $0.974$ \\
	    &     & $6$ &     &     &     &     &     &     &     &     &     & $0$ & $0$ & $47^\mathsf{L}$ \colorthiscell & $3^\mathsf{U}$ \colorthiscell &     &     &     &     &     &     &     & $0.001$ & $0.033$ & $0.972$ \\
	\hline
	    & $2$ & $2$ &     &     &     &     &     &     & $0$ & $0$ & $21^\mathsf{L}$ \colorthiscell & $20$ \colorthiscell & $9$ \colorthiscell & $0$ \colorthiscell & $0^\mathsf{U}$ \colorthiscell &     &     &     &     &     &     &     & $0.000$ & $0.024$ & $0.966$ \\
	    &     & $4$ &     &     &     &     &     &     &     & $0$ & $0$ & $21^\mathsf{L}$ \colorthiscell & $28$ \colorthiscell & $1$ \colorthiscell & $0^\mathsf{U}$ \colorthiscell &     &     &     &     &     &     &     & $0.001$ & $0.035$ & $0.960$ \\
	\hline
	    & $3$ & $2$ &     &     &     &     & $0$ & $0$ & $2^\mathsf{L}$ \colorthiscell & $28$ \colorthiscell & $16$ \colorthiscell & $3$ \colorthiscell & $1$ \colorthiscell & $0$ \colorthiscell & $0^\mathsf{U}$ \colorthiscell &     &     &     &     &     &     &     & $0.000$ & $0.018$ & $0.965$ \\
	\hline
	$44$ & $1$ & $2$ &     &     & $0$ & $0$ & $13^\mathsf{L}$ \colorthiscell & $29$ \colorthiscell & $8^\mathsf{U}$ \colorthiscell &     &     &     &     &     &     &     &     &     &     &     &     &     & $0.000$ & $0.019$ & $0.977$ \\
	    &     & $4$ &     &     &     & $0$ & $0$ & $42^\mathsf{L}$ \colorthiscell & $8^\mathsf{U}$ \colorthiscell &     &     &     &     &     &     &     &     &     &     &     &     &     & $0.001$ & $0.024$ & $0.973$ \\
	    &     & $6$ &     &     &     & $0$ & $0$ & $42^\mathsf{L}$ \colorthiscell & $8^\mathsf{U}$ \colorthiscell &     &     &     &     &     &     &     &     &     &     &     &     &     & $0.001$ & $0.023$ & $0.973$ \\
	\hline
	    & $2$ & $2$ &     & $0$ & $0$ & $14^\mathsf{L}$ \colorthiscell & $32$ \colorthiscell & $4$ \colorthiscell & $0^\mathsf{U}$ \colorthiscell &     &     &     &     &     &     &     &     &     &     &     &     &     & $0.001$ & $0.054$ & $0.957$ \\
	    &     & $4$ &     &     & $0$ & $0$ & $32^\mathsf{L}$ \colorthiscell & $18$ \colorthiscell & $0^\mathsf{U}$ \colorthiscell &     &     &     &     &     &     &     &     &     &     &     &     &     & $0.001$ & $0.043$ & $0.964$ \\
	\hline
	    & $3$ & $2$ & $0$ & $0$ & $6^\mathsf{L}$ \colorthiscell & $38$ \colorthiscell & $6$ \colorthiscell & $0$ \colorthiscell & $0^\mathsf{U}$ \colorthiscell &     &     &     &     &     &     &     &     &     &     &     &     &     & $0.001$ & $0.030$ & $0.948$ \\
\end{tabular}
\end{center}
\end{table*}

\begin{table*}
\caption{Table for $q=101$ and $n=100$. See Section~\ref{ssec:decoding_numerical_results} for the description.}
\label{tab:q=101}
\setlength{\tabcolsep}{1.7pt}	
\begin{tabular}{c|c|c||c|c|c|c|c|c|c|c|c|c|c|c|c|c|c|c|c|c|c|c|c|c|c|c|c|c|c|c|c||c|c|c}
\multicolumn{3}{c||}{Parameters} & \multicolumn{29}{c||}{Number of codes (out of $50$) that have $\tau_\mathsf{max}=$} & \multicolumn{3}{c}{Observed Failure Rates} \\
\,$k$\, & \,$\ell$\, & $\zeta$ & $7$ & $8$ & $9$ & $10$ & $11$ & $12$ & $13$ & $14$ & $15$ & $16$ & $17$ & $18$ & $19$ & $20$ & $21$ & $22$ & $23$ & $24$ & $25$ & $26$ & $27$ & $28$ & $29$ & $30$ & $31$ & $32$ & $33$ & $34$ & $35$ & $\PfTwoBelow$ & $\PfOneBelow$ & $\PfAbove$ \\
\hline \hline
$30$ & $1$ & $2$ &     &     &     &     &     &     &     &     &     &     &     &     &     &     &     &     &     &     &     &     & $0$ & $0$ & $9^\mathsf{L}$ \colorthiscell & $14$ \colorthiscell & $10$ \colorthiscell & $6$ \colorthiscell & $9$ \colorthiscell & $2$ \colorthiscell & $0^\mathsf{U}$ \colorthiscell & $0.000$ & $0.015$ & $0.983$ \\
     &     & $4$ &     &     &     &     &     &     &     &     &     &     &     &     &     &     &     &     &     &     &     &     &     &     & $0$ & $0$ & $13^\mathsf{L}$ \colorthiscell & $20$ \colorthiscell & $13$ \colorthiscell & $4$ \colorthiscell & $0^\mathsf{U}$ \colorthiscell & $0.000$ & $0.012$ & $0.986$ \\
     &     & $6$ &     &     &     &     &     &     &     &     &     &     &     &     &     &     &     &     &     &     &     &     &     &     &     & $0$ & $0$ & $23^\mathsf{L}$ \colorthiscell & $22$ \colorthiscell & $5$ \colorthiscell & $0^\mathsf{U}$ \colorthiscell & $0.000$ & $0.012$ & $0.983$ \\
\hline
     & $2$ & $2$ &     &     &     &     &     &     &     &     &     &     &     &     &     &     &     &     & $0$ & $0$ & $1^\mathsf{L}$ \colorthiscell & $7$ \colorthiscell & $18$ \colorthiscell & $12$ \colorthiscell & $5$ \colorthiscell & $3$ \colorthiscell & $2$ \colorthiscell & $1$ \colorthiscell & $1$ \colorthiscell & $0$ \colorthiscell & $0^\mathsf{U}$ \colorthiscell & $0.000$ & $0.015$ & $0.972$ \\
     &     & $4$ &     &     &     &     &     &     &     &     &     &     &     &     &     &     &     &     &     &     &     & $0$ & $0$ & $2^\mathsf{L}$ \colorthiscell & $18$ \colorthiscell & $19$ \colorthiscell & $5$ \colorthiscell & $4$ \colorthiscell & $2$ \colorthiscell & $0$ \colorthiscell & $0^\mathsf{U}$ \colorthiscell & $0.000$ & $0.000$ & $0.983$ \\
\hline
     & $3$ & $2$ &     &     &     &     &     &     &     &     &     &     &     &     &     & $0$ & $0$ & $0^\mathsf{L}$ \colorthiscell & $11$ \colorthiscell & $9$ \colorthiscell & $12$ \colorthiscell & $5$ \colorthiscell & $4$ \colorthiscell & $6$ \colorthiscell & $2$ \colorthiscell & $0$ \colorthiscell & $1$ \colorthiscell & $0$ \colorthiscell & $0$ \colorthiscell & $0$ \colorthiscell & $0^\mathsf{U}$ \colorthiscell & $0.000$ & $0.023$ & $0.966$ \\
\hline
 $50$ & $1$ & $2$ &     &     &     &     &     &     &     &     &     &     &     &     & $0$ & $0$ & $10^\mathsf{L}$ \colorthiscell & $15$ \colorthiscell & $18$ \colorthiscell & $7$ \colorthiscell & $0^\mathsf{U}$ \colorthiscell &     &     &     &     &     &     &     &     &     &     & $0.000$ & $0.016$ & $0.987$ \\
     &     & $4$ &     &     &     &     &     &     &     &     &     &     &     &     &     & $0$ & $0$ & $10^\mathsf{L}$ \colorthiscell & $25$ \colorthiscell & $15$ \colorthiscell & $0^\mathsf{U}$ \colorthiscell &     &     &     &     &     &     &     &     &     &     & $0.000$ & $0.012$ & $0.985$ \\
     &     & $6$ &     &     &     &     &     &     &     &     &     &     &     &     &     & $0$ & $0$ & $3^\mathsf{L}$ \colorthiscell & $28$ \colorthiscell & $19$ \colorthiscell & $0^\mathsf{U}$ \colorthiscell &     &     &     &     &     &     &     &     &     &     & $0.000$ & $0.015$ & $0.985$ \\
\hline
     & $2$ & $2$ &     &     &     &     &     &     &     &     &     & $0$ & $0$ & $6^\mathsf{L}$ \colorthiscell & $22$ \colorthiscell & $9$ \colorthiscell & $10$ \colorthiscell & $1$ \colorthiscell & $2$ \colorthiscell & $0$ \colorthiscell & $0^\mathsf{U}$ \colorthiscell &     &     &     &     &     &     &     &     &     &     & $0.000$ & $0.014$ & $0.973$ \\
     &     & $4$ &     &     &     &     &     &     &     &     &     &     &     & $0$ & $0$ & $13^\mathsf{L}$ \colorthiscell & $22$ \colorthiscell & $12$ \colorthiscell & $3$ \colorthiscell & $0$ \colorthiscell & $0^\mathsf{U}$ \colorthiscell &     &     &     &     &     &     &     &     &     &     & $0.000$ & $0.012$ & $0.981$ \\
\hline
     & $3$ & $2$ &     &     &     &     &     &     &     & $0$ & $0$ & $7^\mathsf{L}$ \colorthiscell & $18$ \colorthiscell & $12$ \colorthiscell & $8$ \colorthiscell & $4$ \colorthiscell & $1$ \colorthiscell & $0$ \colorthiscell & $0$ \colorthiscell & $0$ \colorthiscell & $0^\mathsf{U}$ \colorthiscell &     &     &     &     &     &     &     &     &     &     & $0.000$ & $0.023$ & $0.977$ \\
\hline
 $70$ & $1$ & $2$ &     &     &     & $0$ & $0$ & $23^\mathsf{L}$ \colorthiscell & $18$ \colorthiscell & $9$ \colorthiscell & $0^\mathsf{U}$ \colorthiscell &     &     &     &     &     &     &     &     &     &     &     &     &     &     &     &     &     &     &     &     & $0.000$ & $0.016$ & $0.983$ \\
     &     & $4$ &     &     &     &     & $0$ & $0$ & $35^\mathsf{L}$ \colorthiscell & $15$ \colorthiscell & $0^\mathsf{U}$ \colorthiscell &     &     &     &     &     &     &     &     &     &     &     &     &     &     &     &     &     &     &     &     & $0.000$ & $0.012$ & $0.984$ \\
     &     & $6$ &     &     &     &     & $0$ & $0$ & $30^\mathsf{L}$ \colorthiscell & $20$ \colorthiscell & $0^\mathsf{U}$ \colorthiscell &     &     &     &     &     &     &     &     &     &     &     &     &     &     &     &     &     &     &     &     & $0.000$ & $0.015$ & $0.986$ \\
\hline
     & $2$ & $2$ &     & $0$ & $0$ & $5^\mathsf{L}$ \colorthiscell & $24$ \colorthiscell & $15$ \colorthiscell & $6$ \colorthiscell & $0$ \colorthiscell & $0^\mathsf{U}$ \colorthiscell &     &     &     &     &     &     &     &     &     &     &     &     &     &     &     &     &     &     &     &     & $0.000$ & $0.021$ & $0.974$ \\
     &     & $4$ &     &     & $0$ & $0$ & $0^\mathsf{L}$ \colorthiscell & $35$ \colorthiscell & $15$ \colorthiscell & $0$ \colorthiscell & $0^\mathsf{U}$ \colorthiscell &     &     &     &     &     &     &     &     &     &     &     &     &     &     &     &     &     &     &     &     & $0.000$ & $0.008$ & $0.984$ \\
\hline
     & $3$ & $2$ & $0$ & $0$ & $16^\mathsf{L}$ \colorthiscell & $29$ \colorthiscell & $4$ \colorthiscell & $1$ \colorthiscell & $0$ \colorthiscell & $0$ \colorthiscell & $0^\mathsf{U}$ \colorthiscell &     &     &     &     &     &     &     &     &     &     &     &     &     &     &     &     &     &     &     &     & $0.000$ & $0.013$ & $0.965$ \\
\end{tabular}
\end{table*}

\appendices

\section{Efficient Decoding}\label{app:efficiently_decoding}

In this appendix, we show how to implement the decoder in Section~\ref{sec:decoding} more efficiently than solving a linear system. The bottleneck of \cref{alg:decoder} is Line~\ref{line:solve_linearized_problem}, which solves \cref{prob:pade_approx}. We first show how to solve it fast using row reduction.

The following theorem shows how a minimal solution of \cref{prob:pade_approx} can be found and in which complexity.
We need the following well-known notation.
For a vector $\m \in \Fq[X]^r$, and a \emph{shift vector} $\s \in \ZZ^r$, we define its \emph{$\s$-shifted degree} as $\deg_\s \m := \max_j\{\deg m_j+ s_j\}$ and the \emph{$\s$-pivot of $\m$} to be the right-most index $i$ such that $\deg m_i+ s_i = \deg_\s \m$.
A matrix $\M \in \Fq[X]^{r \times r}$ is in \emph{$\s$-shifted weak Popov form} if all its rows have distinct $\s$-pivots.
It is well-known that a matrix in $\s$-shifted weak Popov form is $\s$-row reduced, i.e., for all $i$, its row with $\s$-pivot $i$ has minimal $\s$-shifted degree among all non-zero elements in the matrix' row space of $\s$-pivot $i$.
Furthermore, any full-rank square matrix $\ModuleBasis \in \Fq[X]^{r \times r}$ can be transformed (by preserving its row space) into $\s$-shifted weak Popov form using the Las-Vegas algorithm in \cite{neiger2016fast} with complexity $O^\sim(r^\omega \deg \ModuleBasis)$ operations in $\Fq$,
where $\omega$ is the matrix multiplication exponent and $\deg \ModuleBasis$ denotes the maximal degree of $\ModuleBasis$.
Note that the algorithm in \cite{neiger2016fast} ouputs an \emph{$\s$-shifted Popov form}, which is in particular in $\s$-shifted weak Popov form.

\begin{theorem}\label{thm:decoder_computation_and_complexity}
	Consider an instance of \cref{prob:pade_approx}.
	Let
	\begin{align*}
	\ModuleBasis := \begin{bmatrix}
	\I_{|\Iset_{\decParam+1}| \times |\Iset_{\decParam+1}|} & \A \\
	& G \cdot \I_{|\Iset_\decParam| \times |\Iset_\decParam|}
	\end{bmatrix}, %
	\end{align*}
	where $\A \in \Fq[X]^{|\Iset_{\decParam+1}| \times |\Iset_{\decParam}|}$ is a matrix whose $(\i,\j)$-th entry for $\i \in \Iset_{\decParam+1}$ and $\j \in \Iset_{\decParam}$ (fix orders $\i_i$ and $\j_j$ of the elements in $\Iset_{\decParam+1}$ and $\Iset_{\decParam}$, respectively, both starting with $\0$) is
	\begin{align*}
	\A_{\i,\j} := \begin{cases}
	R, &\text{if } \i = \j, \\
	-\eta_\mu X^{k-1+t_\mu}, &\text{if } \i = \j+\deltamu \text{ for some $\mu$}, \\
	0, &\text{else}.
	\end{cases}
	\end{align*}
	Furthermore, let $\s \in \ZZ^{|\Iset_{\decParam+1}|+|\Iset_{\decParam}|}$ be such that
	\begin{align*}
	s_i = \begin{cases}
	k, &\text{if } i=1, \\
	k-1, &\text{if } 1 < i \leq |\Iset_{\decParam+1}|, \\
	0, &\text{if } |\Iset_{\decParam}| < i \leq |\Iset_{\decParam+1}|+|\Iset_{\decParam}|.
	\end{cases}
	\end{align*}
	
	Let $\ModuleBasis'$ be a basis in $\s$-shifted weak Popov form of the module spanned by the rows of $\ModuleBasis$, and let $\m$ be the (unique) row of $\ModuleBasis'$ with $\s$-pivot $1$.
	Then,
	\begin{align*}
	[\lambda_{\i_1}, \lambda_{\i_2}, \dots, \lambda_{\i_{|\Iset_{\decParam+1}|}}, \psi_{\j_1},\dots,\psi_{\j_{|\Iset_{\decParam}|}}] := \m
	\end{align*}
	is a solution of \cref{prob:pade_approx} of minimal $\deg \lambda_\0$.
	
	The matrix $\ModuleBasis'$ can be computed using the Las-Vegas algorithm in \cite{neiger2016fast} in
	\begin{align*}
	O^\sim \!\left( \big(|\Iset_{\decParam+1}|+|\Iset_{\decParam}|\big)^\omega n \right) \subseteq O^\sim\!\left( \left(e\tfrac{\ell+\zeta+1}{\ell}\right)^{\ell \omega} n\right)
	\end{align*}
	operations over $\Fq$, where $e$ is Euler's constant.
\end{theorem}

\begin{proof}
	We first show that the rows of $\ModuleBasis$ form a basis of the module $\Module$ of vectors
	\begin{align*}
	\v := [\lambda_{\i_1}, \lambda_{\i_2}, \dots, \lambda_{\i_{|\Iset_{\decParam+1}|}}, \psi_{\j_1},\dots,\psi_{\j_{|\Iset_{\decParam}|}}] \in \Fq[X]^{|\Iset_{\decParam+1}|+|\Iset_{\decParam}|}
	\end{align*}
	that satisfy the congruence relation in \eqref{eq:pade_deg_lambda}. Consider an element $\v \in \Module$. Then, there are polynomials $\chi_\i \in \Fq[X]$, for $\i \in \Iset_{\decParam}$, such that
	\begin{align*}
	\psi_\i = \lambda_\i R - \sum_{\mu=1}^{\ell} \lambda_{(\i+\deltamu)} \eta_\mu X^{k-1+t_\mu} + \chi_\i G
	\end{align*}
	for all $\i \in \Iset_{\decParam}$. By the choice of $\ModuleBasis$, we have
	\begin{align*}
	\v = [\lambda_{\i_1}, \lambda_{\i_2}, \dots, \lambda_{\i_{|\Iset_{\decParam+1}|}}, \chi_{\j_1},\dots,\chi_{\j_{|\Iset_{\decParam}|}}] \cdot \ModuleBasis.
	\end{align*}
	Hence, $\Module$ is contained in the row space of $\ModuleBasis$. 
	On the other hand, any row of $\ModuleBasis$ is in $\Module$ since:
	\begin{itemize}
		\item Index the first $|\Iset_{\decParam+1}|$ rows of $\ModuleBasis$ by $\i$ and the $|\Iset_{\decParam}|$ congruence relations in \eqref{eq:pade_deg_lambda} by $\j$. Then row $\i$ satisfies relation $\j$ since
		\begin{align*}
		R &= R + \sum_{\mu=1}^{\ell} 0 \cdot \eta_\mu X^{k-1+t_\mu}, &&\text{if } \j = \i, \\
		0  &= -\eta_\mu X^{k-1+t_\mu} + \eta_\mu X^{k-1+t_\mu}, &&\text{if } \j = \i-\deltamu, \\
		0  &= 0 + \sum_{\mu=1}^{\ell} 0 \cdot \eta_\mu X^{k-1+t_\mu}, &&\text{else}.
		\end{align*}
		\item The last $|\Iset_{\decParam}|$ rows of $\ModuleBasis$ satisfy \eqref{eq:pade_deg_lambda} since
		\begin{align*}
		0 &\equiv G + \sum_{\mu=1}^{\ell} 0 \cdot \eta_\mu X^{k-1+t_\mu} \mod G.
		\end{align*}
	\end{itemize}
	Also, the rows of $\ModuleBasis$ are linearly independent since the matrix is in upper triangular form with non-zero diagonal entries.
	
	Since the row $\m$ of $\ModuleBasis'$ has $\s$-pivot $1$, the degree inequalities \eqref{eq:pade_deg_lambda} and \eqref{eq:pade_deg_psi} are fulfilled. This is true since, by the definition of the $\s$-pivot and the choice of the shift, we have
	\begin{align*}
	\deg \lambda_{\0} + k_1 &= \deg \lambda_{\i_1} + k > \deg \lambda_{\i_i} + k-1 \\
	\Leftrightarrow \deg \lambda_{\i_i} &\leq \deg \lambda_{\0}
	\end{align*}
	for all $i=1,\dots,|\Iset_{\decParam+1}|$, and
	\begin{align*}
	\deg \lambda_{\0} + k_1 &= \deg \lambda_{\i_1} + k > \deg \lambda_{\j_j}  \\
	\Leftrightarrow \deg \lambda_{\j_j} &\leq \deg \lambda_{\0} +k-1
	\end{align*}
	for all $j=1,\dots,|\Iset_{\decParam}|$. Hence, $\m$ is a solution of Problem~\ref{prob:pade_approx}.
	Moreover, it is also one of minimal degree since $\ModuleBasis'$ is $\s$-row reduced, i.e., $\m$ has minimal $\s$-shifted degree among all non-zero vectors in the row space of $\ModuleBasis'$ with $\s$-pivot $1$.
	
	As for the complexity, the matrix $\ModuleBasis$ has $|\Iset_{\decParam+1}|+|\Iset_{\decParam}|$ rows and columns, and maximal degree at most $n$. The complexity follows by the algorithm in \cite{giorgi2003complexity}, see complexity expression above the theorem.
\end{proof}

The other operations of \cref{alg:decoder} are standard polynomial operations, which all have complexity $O^\sim(n)$ operations in $\Fq$, see, e.g., \cite{von2013modern}: 
\begin{itemize}
	\item $R$ in \cref{line:R_and_G} is obtained via Lagrange interpolation with $n$ point tuples,
	\item $G$ in \cref{line:R_and_G} can be computed via a subproduct tree,
	\item Lines~\ref{line:divisibility_check} and \ref{line:division} can be implemented by a division with remainder,
	\item and \cref{line:evaluation} is a multi-point evaluation of a polynomial of degree $<n$ at $n$ points.
\end{itemize}
Hence, the bottleneck is Line~\ref{line:solve_linearized_problem}, and Algorithm~\ref{alg:decoder} can be implemented with complexity
\begin{align*}
O^\sim\!\left( \left(e\tfrac{\ell+\zeta+1}{\ell}\right)^{\ell \omega} n\right)
\end{align*}
operations in $\Fq$, where $e$ is Euler's number.

\begin{IEEEbiographynophoto}{Peter Beelen}
received his Master’s degree in Mathematics from the University of Utrecht, The Netherlands, in 1996. In 2001 he received his Ph.D. degree in Mathematics from the Technical University of Eindhoven, The Netherlands. Since October 2004 he has been a staff member of the Technical University of Denmark (DTU), Kongens Lyngby, Denmark. He has been an assistant professor at DTU till January 2007 and an associate professor till August 2014. Since September 2014 he has worked at DTU as professor. His research interests include various aspects of algebra and its applications, notably algebraic curves and algebraic coding theory.
\end{IEEEbiographynophoto}

\begin{IEEEbiographynophoto}{Sven Puchinger}
(S'14, M'19) received the B.Sc. degree in electrical engineering and the B.Sc. degree in mathematics from Ulm University, Germany, in 2012 and 2016, respectively. During his studies, he spent two semesters at the University of Toronto, Canada. He received his Ph.D. degree from the Institute of Communications Engineering, Ulm University, Germany, in 2018. He has been a postdoctoral researcher with the Technical University of Munich, Germany, from 2018 to 2019, at the Technical University of Denmark, Denmark, from 2019 to 2021, and again at the Technical University of Munich in 2021. Since 2021, he is with Hensoldt Sensors GmbH, Ulm, Germany. His research interests are coding theory, its applications, and related computer-algebra methods.
\end{IEEEbiographynophoto}

\begin{IEEEbiographynophoto}{Johan Rosenkilde}
is a Research Engineer at GitHub since 2021. Before that he
was at the Technical University of Denmark, first as assistant professor, then
as associate professor. He holds a Master's degree in computer science and a PhD
in mathematics from the same university, and was a post-doc at both Ulm
University, Germany and Inria, France. His algebraic research interests include
coding theory and computer algebra.
\end{IEEEbiographynophoto}

\end{document}